\documentclass[10pt, conference, letterpaper]{IEEEtran}

\usepackage{amsmath,amssymb,amsfonts,amsthm}
\usepackage{algorithmic}
\usepackage{graphicx}
\usepackage{textcomp}

\usepackage[linesnumbered,ruled]{algorithm2e}
\usepackage{pifont}
\usepackage{gensymb}
\usepackage{booktabs}
\usepackage{capt-of}
\usepackage{thmtools}
\usepackage{thm-restate}
\usepackage{adjustbox}
\usepackage{makecell}
\usepackage{enumitem}
\usepackage{floatrow}
\usepackage{tabularx}
\usepackage[style=ieee, sorting=nty, maxbibnames=99]{biblatex}
\usepackage{flushend}

\newfloatcommand{capbtabbox}{table}[][\FBwidth]

\newtheorem{theorem}{Theorem}[section]
\newtheorem{corollary}[theorem]{Corollary}

\newtheorem{definition}[theorem]{Definition}

\newcommand{\cmark}{\ding{51}}
\newcommand{\xmark}{\ding{55}}

\newcommand{\squishlist}
{
    \begin{list}{$\bullet$}
    {
        \setlength{\itemsep}{0pt}      \setlength{\parsep}{3pt}
        \setlength{\topsep}{3pt}       \setlength{\partopsep}{0pt}
        \setlength{\leftmargin}{1.5em} \setlength{\labelwidth}{1em}
        \setlength{\labelsep}{0.5em}
    }
}

\newcommand{\squishend}
{
    \end{list}
}

\newcommand{\squishenumerate}
{
    \begin{enumerate}{}
    {
        \setlength{\itemsep}{0pt}      \setlength{\parsep}{3pt}
        \setlength{\topsep}{3pt}       \setlength{\partopsep}{0pt}
        \setlength{\leftmargin}{1.5em} \setlength{\labelwidth}{1em}
        \setlength{\labelsep}{0.5em}
    }
}
\newcommand{\squishenumerateend}
{
    \end{enumerate}
}

\newif\ifknoxrestating
\providecommand{\IfRestatedTF}[2]{\ifknoxrestating #1\else #2\fi}

\usepackage[
    hidelinks,
    breaklinks=true,
    pdftitle={Knox: Fortifying Smart Spaces With Safety Guarantees},
    pdfauthor={Rishabh Menezes, Jadon T. Schuler, Kaimeng Zhu, Oliver Rogalski, Indranil Gupta},
    pdfsubject={Distributed, Parallel, and Cluster Computing},
    pdfkeywords={IoT, smart spaces, safety verification, static analysis}
]{hyperref}

\IEEEoverridecommandlockouts

\def\BibTeX{{\rm B\kern-.05em{\sc i\kern-.025em b}\kern-.08em
    T\kern-.1667em\lower.7ex\hbox{E}\kern-.125emX}}
\begin{document}

\title{Knox: Fortifying Smart Spaces With Safety Guarantees}

\author{
    \IEEEauthorblockN{Rishabh Menezes, Jadon T. Schuler, Kaimeng Zhu, Oliver Rogalski, Indranil Gupta}
    \IEEEauthorblockA{University of Illinois at Urbana-Champaign}
    \IEEEauthorblockA{\{menezes4, jadonts2, kaimeng2, oliverr3, indy\}@illinois.edu}
    \thanks{Accepted for publication in the Proceedings of the 2026 IEEE
    International Symposium on Reliable Distributed Systems (SRDS). This is the
    authors' extended version, which additionally includes the full proofs of
    all theorems in the appendix.
    \copyright~2026 IEEE. Personal use of this material is permitted.
    Permission from IEEE must be obtained for all other uses, in any current or
    future media, including reprinting/republishing this material for
    advertising or promotional purposes, creating new collective works, for
    resale or redistribution to servers or lists, or reuse of any copyrighted
    component of this work in other works.}
}

\maketitle

\begin{abstract}
Internet of Things (IoT) devices in smart spaces and buildings are an emerging class of distributed systems with critical safety requirements. This paper presents {\it Knox}, the first system to enable safety checking in IoT-enabled smart spaces. 
Knox's contributions include (i) {\it safety specifications}: a new language for safety clauses in such smart spaces, and (ii) {\it static safety checking}: two new algorithms for static verification of multiple safety properties across multiple routines running inside a smart space. Since the latter problem is NP-hard,  we present and analyze novel and explainable algorithms for the static version of the problem. We also present optimizations that further reduce runtime. Our analysis and experimental results with real datasets show that Knox reduces checking time significantly compared to baselines, while providing high accuracy in catching safety violations.
\end{abstract}

\section{Introduction}
\label{sec:intro}

In a smart space---a home, building, or campus---Internet of Things (IoT) devices from multiple vendors cover all aspects of operation, including smart plugs, windows, doors, kitchen appliances, gardening, and cleaning. The smart home market alone
is expected to grow from \$127 B today to over \$500 B by 2030 \cite{GVR26}, %
with an expected 39 Billion IoT devices~\cite{IotA25} from 3000+ vendors~\cite{IotA24}.

Devices are coordinated by a central hub, such as Alexa~\cite{AmazonAlexa}, Google Home~\cite{GoogleHome}, or SmartThings~\cite{SamsungSmartThings}, which (i) communicates with IoT devices through the wireless network, and (ii)  launches user-programmed {\it automations}.
The most common form of automation today is  %
called a {\it routine} \cite{GoogleHome, AmazonAlexa, AppleHome, SamsungSmartThings, IFTTT}.

A routine is a (small) program containing a sequence of commands, each executing one operation on one device; thus a routine may access different devices.
Once {\it submitted} by a user, routines are held by the hub, and each is {\it activated} via one or more {\it triggers}: time, a sensor (e.g., thermostat setting crossing a threshold temperature), or the actions of another routine. Any routine may also be manually triggered at any time by the user via an app, their voice, etc. Example routines are:

{
\squishlist
\small
\item\texttt{{\bf R1:} $\langle$Trigger: time = 8pm$\rangle$ outdoorCamera = ON; door = LOCKED; porchLight = ON;}
\item\texttt{{\bf R2:} $\langle$time = 1am$\rangle$ porchLight = OFF;}
\item\texttt{{\bf R3:} $\langle$no motion for 10 mins$\rangle$ floodLight=OFF;}
\squishend
}

When a routine is triggered, the hub sequentially launches its %
commands
(IoT devices themselves do not manage routines, as device firmware is immutable). %
As routines may be triggered by time, sensors, or user input,
it is common for arbitrarily many routines to be running {\it concurrently} inside a smart space. This creates a major reliability challenge
that a smart space as a whole may enter {\it unsafe and unreliable
states}~\cite{InternetOfShit}. Maintaining reliability in a smart space requires that users have the ability to specify {\textbf{\textit{safety expectations}}}
 that  {\textbf{\textit{hold across the space at all times}}}, i.e., common concerns regardless of which routines are running. These safety expectations range
from life-critical (e.g., camera's field of view must be lit, or exhaust fan must be on whenever oven is on)
to everyday concerns (e.g., TV must not come on during an ongoing Zoom call). %
\textbf{\textit{It is unreasonable to hope that users will carefully architect each routine to satisfy their myriad safety expectations}}; safety clauses {\it must} therefore be specified separately from routines. %

Unfortunately, today's automation systems are best-effort only. They neither allow users to specify safety clauses nor detect violations. Google Home states that ``routines are for convenience only, not safety- or security-critical use cases''~\cite{GoogleSafetyCritical}, and Samsung SmartThings recommends avoiding devices %
``which could cause damage to or loss of any property'' \cite{TermsSmartThings}. Yet smart spaces {\it are} safety critical, and
incidents
continue to violate users'  reliability expectations~\cite{HomekitForum1, HomekitForum2, HomekitForum3, HomekitForum4}, e.g., a failed automation overheated a room to $40\degree$C \cite{SmartThingsForum1}.

This paper presents the theoretical foundations for {\it Knox}, {\textbf{\textit{the first system to ensure {\underline{static}} safety in smart environments}}}, allowing both specification and {\it automatic static checking} of an arbitrary number of safety clauses under an arbitrary number of routines and devices. An example safety clause in Knox is:

\noindent {\small
\texttt{S1}: \texttt{IF (outdoorCamera == ON)  \\ THEN (porchLight == ON OR floodLight == ON)}. 
}

At night, this safety clause ensures that whenever the camera is on, there is sufficient light in its field of view. Yet, since outdoorCamera, porchLight, and floodLight can each be switched on or off by other concurrent routines, the clause may be violated: %
given routines R1, R2, and R3 above,  S1 is violated at 1 am.
If {\it any} combination of a set of routines has the possibility of violating a safety clause, the system is unsafe and {\it should be detected as such when R1, R2, R3, and S1 are submitted} to the hub.
Today, there are no methods for specifying S1, nor algorithms for the hub to detect if a given set of routines could ever violate it.

In this paper we (I) propose a new grammar for specifying safety clauses in a smart space and (II) design and analyze algorithms for statically checking a set of safety clauses for a given set of routines, reporting all safety violations. Our system, Knox, is \textbf{\textit{safe}}: it does not miss any safety violations. %
Concretely, under (II), we need to solve three problems: 

\noindent 1. \textbf{{Feasibility}}: Existence of at least one space state where all safety clauses are satisfied. 

\noindent 2. \textbf{{Single Routine Safety}}: When a new routine is first written and submitted to the hub, that routine itself does not violate any safety clauses. %

\noindent 3. \textbf{{Concurrent Routine Safety}}: Given a set of submitted routines at the hub, find all subgroups of routines with at least one interleaving that violates any safety clause.

An additional requirement is {\textit{Explainability}}: for detected violations of either Single Routine Safety or Concurrent Routine Safety, we should pinpoint the exact commands that violated a given safety clause. As misunderstandings of concurrent program behavior in IoT environments are common among novice-programmer users~\cite{BugsinTAP, SmartDevicesStupid, EUDebug, iotacalc}, this feedback helps identify and fix culprit routines.

We solve the \textbf{\textit{static versions}} of these problems where all checks run %
when a new routine or a new safety clause is submitted to the hub.
This has the advantage of securing all static safety properties before runtime, so no dynamic checks of the same properties are needed. {(Dynamic issues, including failures, are beyond our current scope and are left for future work.)} %
First, the smart space's devices may be in arbitrary initial states, making safety checking (2) \& (3) %
hard.
Second, the state space is finite, but massive. In a smart space with only 30 devices, each of which can be in one of 4 states (e.g., HI, MED, LO, OFF), there
are {\textbf{\textit{over $10^{18}$ unique  states}}}! This means manual checking is untenable, and even automated brute-force solutions are too slow.
Model checkers specialized for infinite-state systems are also unsuitable \cite{DURAN2020100497}.

Feasibility (1) turns out to be easy to solve: an SMT (Satisfiability Modulo Theories)-based solver suffices.
However, for Single Routine Safety (2) and Concurrent Routine Safety (3), large state spaces make SMT-based approaches prohibitively expensive---with $R$ routines each containing $k$ instructions, there are {\textbf{\textit{$\frac{(Rk)!}{(k!)^R}$ possible interleavings}}}! Naively checking safety for each command pair of every interleaving runs SMT $(\frac{(Rk)!}{(k!)^R}\cdot Rk)$ times. Our new techniques are instead guided by the principle of {\textbf{\textit{erring on the side of safety}}}: they detect all safety violations and never miss any (zero false negatives), though false positives may occur, i.e., Knox may flag a few cases that are not actually violations (we measure this).

The contributions of this paper are:
\begin{flushleft}
\squishenumerate
\item A new expressive grammar for specifying safety clauses in IoT settings (Section \ref{sec:grammar}).
\item New fast algorithms for Single and Concurrent Routine Safety checking  %
(Section \ref{sec:Algorithms}).
\item Two new optimizations to reduce runtime (Section \ref{sec:opt}).
\item Formal proofs of correctness (Sections  \ref{sec:Algorithms} \& \ref{sec:opt}). %
\item Experiments to measure speed \& accuracy,  comparison 
vs. a baseline %
(Section \ref{sec:eval}).
\squishenumerateend
\end{flushleft}

To the best of our knowledge, Knox provides the first fully-specified formal grammar and associated safety checking algorithms for efficient static verificati
on of concurrent automations in smart space environments.

\section{System Model}

Knox makes the following system assumptions. %

{\bf Physical Components}: The smart space consists of a single centralized hub and a series of devices: %
\squishlist
\item Devices have different types of state --- including binary (on, off), multi (colors, operation modes), or numerical (light intensity, temperature).
\item Each device can receive and execute device-specific commands. E.g., for a smart plug: \{ON, OFF\}. For a sprinkler, \{ON, OFF, flow:0-100\}.
\item A centralized hub stores and initiates routines (by time, sensor value, manual, etc.) that control any present devices by issuing commands. Examples are Alexa, Google Home Hub, Samsung Smart Things hub, etc.
\item Neither the hub nor the devices fail. We will show that even without failures, which we leave to future work, safety violations pose a challenging problem.
\squishend

{\bf Logical Components}: These assumptions define command execution and organization:
\squishlist
\item Commands are device specific and may be short (e.g., light on) or long (e.g., preheat oven, open garage door, etc.). 
\item Routines contain a finite sequence of commands, without any conditionals or branches (this is consistent with routines in Alexa, Google Home, etc.).
\item A device's state may only be modified by a command. %
\item One or more submitted routines exist, held at the hub. 
\item Once triggered, a routine's commands execute sequentially. %
\item A routine may be activated (triggered) at any arbitrary time. Multiple routines may be active simultaneously (including multiple copies of the same routine).
\squishend

We define the {\it state space} as the collective mapping of all devices to their respective states. %

\section{Related Work}

 Classical human-computer interaction work indicates that smart space application design must anticipate %
 conflicting goals: smart space software should help users ``manage their lives rather than manage individual devices''~\cite{Davidoff}.
Table~\ref{qual-compare} summarizes the closest related work w.r.t. our goals: static and dynamic safety checks, an expressive safety language, and support for routines. To the best of our knowledge, Knox is the first to satisfy all four (its static checking also ensures no dynamic conflicts).

\begin{table}
    \centering\small
    \begin{tabular}{ccccccc}
        \toprule
        Auditor & Static & Dynamic & Expressive & Routines \\
        \midrule
        SafeHome \cite{Ahsan01HotEdge} & \xmark & \xmark & \cmark\ & \xmark \\
        APEX \cite{Zhou} & \xmark & \cmark & N/A & \xmark \\
        SIFT \cite{Liang} & \cmark & \xmark & \xmark & \xmark \\
        TapChecker \cite{TapChecker} & \cmark & \xmark & \cmark & \xmark \\
        Knox & \cmark & \small implied %
        & \cmark & \cmark \\
        \bottomrule
    \end{tabular}
    \caption{\small\bf Related Systems Comparison: \textnormal{We compare safety checkers by the properties they consider: static (submission time) or dynamic (run time). We also consider how expressive their rules can be and their support for routine programming.}}
    \label{qual-compare}
    \vspace{-3.5mm}
\end{table}

APEX~\cite{Zhou} targets safety in a smart space, accepting user-submitted {\it preconditions} which it  satisfies before executing a command. But APEX fails to consider contrapositives. Consider 2  preconditions:
1. (\texttt{AC ON} $\implies$ \texttt{Window CLOSED}), and
2. (\texttt{Bleach counters} $\implies$ \texttt{Window OPEN})\footnote{In APEX, the LHS is a precondition for executing the RHS. We instead normalize arrow directions here for exposition.}.
APEX treats these two rules independently. In turning on the AC, APEX  first automatically closes the windows, even if there are  bleach fumes in the kitchen!
Not only does this  \textit{allow} rule violations, it may \textit{actively create} new violations.

\begin{figure}
    \centering
    \includegraphics[width=\columnwidth]{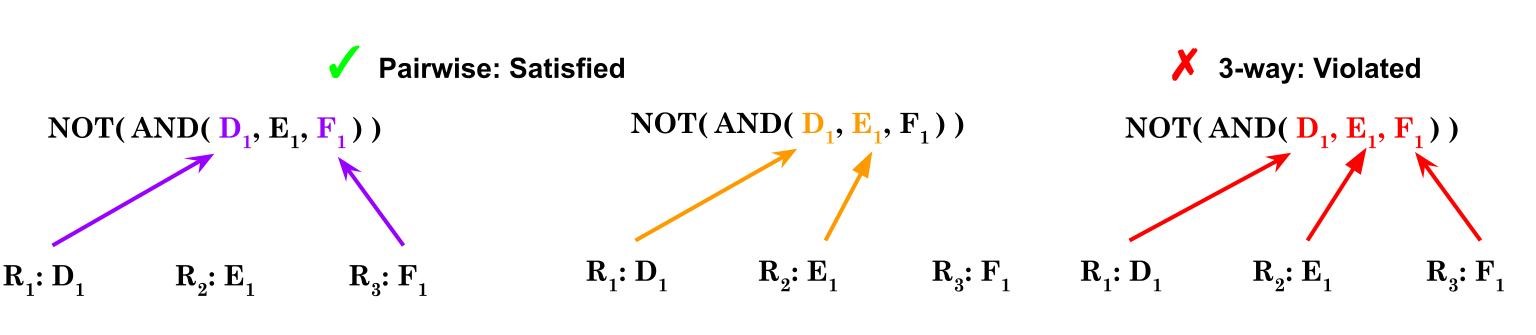}
    \caption{\small\bf Incompleteness of pairwise concurrent routine verification in TapChecker \cite{TapChecker}: \textnormal{Based on a given safety rule's semantics, some violations between 3 or more routines may never be caught. In the example, pairwise comparisons between $R_1, R_2,$ and $R_3$ will consider the overall rule satisfied, but will miss the conflict when all three run concurrently.}} %
    \medskip
    \small

    \label{fig:tapcheckerconflict}
\end{figure}

SIFT~\cite{Liang}  targets safety, but limits routines to a single command and its safety rules to conditionals, %
making it less expressive than Knox.
TapChecker~\cite{TapChecker} uses SMT verification to detect various types of conflicts in trigger-action programming (TAP) statements, with %
SMT-expressible safety rules.
As  SMT expressions only capture a TAP statement's  end result,
TapChecker
cannot analyze multi-command routines or conflicts  from concurrent interleavings, and %
it checks routines only pairwise, potentially allowing rule violations between 3 or more routines (Figure~\ref{fig:tapcheckerconflict}). Hence, while TapChecker never marks safe routines as unsafe (no false positives),
it misses many potential violations (false negatives). Knox never misses any violations but has occasional false positives. Erring on the side of safety is more practical for real deployments, as Section~\ref{sec:eval}'s experiments show.

IOTA~\cite{iotacalc} is a formal calculus for modeling concurrent devices and automations to specify IoT programs.
IOTA neither considers user-provided safety rules,
nor optimizes conflict search beyond reducing predicates on numerical device state to boolean expressions, so its safety checking suffers from  combinatorial interleaving, which Knox addresses.
SafeHome~\cite{Ahsan02Eurosys,Ahsan01HotEdge} features an ``Eventual Consistency'' scheduling algorithm for concurrent routines, %
but does not design or implement safety. SafeHome can use Knox orthogonally.

Formal verification for distributed systems has flourished~\cite{Hance, Hawblitzel}.  However, smart space users are lay users and cannot be expected to learn formal verification tools. 

\vspace{-2mm}

\section{The Safron Grammar}
\label{sec:grammar}

\begin{table}
{
    \small
    \centering
    \caption{\bf Examples of Safron Safety Rules}
    \label{table::rule-examples}
    \begin{tabularx}{\linewidth}{| >{\raggedright\arraybackslash}X | l |}
        \toprule
        {\bf Safety Rule} & {\bf Explanation} \\
        \midrule
        1. \texttt{thermostat>60 AND thermostat<80} & Regulate Temp. Range \\
        \midrule
        2. \texttt{AT MOST 2 (sprinkler1==ON, sprinkler2==ON, sprinkler3==ON)} & Insufficient Water Pressure \\
        \midrule
        3. \texttt{IF window==OPEN THEN AC==OFF} & Energy Savings\\
        \midrule
        4. \texttt{IF (backyardGrill==ON OR backyardSmoker==ON) THEN backyardWindow==CLOSED} & Prevent indoor smoke \\
        \bottomrule
    \end{tabularx}
    \vspace{-2mm}
}
\end{table}
\vspace{2mm}

\begin{figure}
{\small
\begin{verbatim}
RULE    := BINOP | IMP | ELEM
BINOP   := ELEM AND ELEM | ELEM OR ELEM
IMP     := IF ELEM THEN ELEM ELSE ELEM
           | IF ELEM THEN ELEM
SINGLE  := ANY(LIST) | ALL(LIST)
           | AT LEAST NUM (LIST)
           | AT MOST NUM (LIST)
           | EXACTLY NUM (LIST) | !ELEM | ATOM
NUM     := [1-9][0-9]*
LIST    := RULE, LIST | RULE
ELEM    := SINGLE | (RULE)
ATOM    := DeviceID.StateID OP value
OP      := == | < | > | <= | >=
\end{verbatim}
}
\vspace{-2mm}
\caption{\textbf{Safron: Grammar for \underline{Saf}ety Among \underline{Ro}uti\underline{n}es}}
\label{fig:Safron}
\end{figure}

Knox allows safety clauses to be written using our \textbf{Safron} (\underline{Saf}ety Among \underline{Ro}uti\underline{n}es) grammar, defined in Figure~\ref{fig:Safron}.
Safety rules can come prebaked with the smart space or alternatively can be added at any time by a user. Admitted rules are held at the hub.  
The grammar is quantifier-free and time-free, as rules are expected to hold at all times.

Table~\ref{table::rule-examples} shows  example rules written in Safron. %
A base (ATOM) Safron clause expresses a boolean predicate on a single device's state. Operations include standard boolean connectors such as NOT (!), AND, OR, and IF conditions, 
with ANY and ALL aggregate statements serving as syntactic sugar. $k$-constraint operators such as AT LEAST, AT MOST, and EXACTLY are common for smart spaces.  %

\noindent \emph{Tree Representation of Safety Rules}:
We auto-translate each rule into a \emph{parse tree}, first by rewriting (Table \ref{parse-modifications}) and then creating a tree (Figure \ref{fig:RuleParseTree}). %
The leaves represent atomic device states (ATOM rule) and internal nodes are logical connectors (e.g.,  AND). Edges are subexpressions with  logical connectors. %

\thisfloatsetup{subfloatrowsep=none}
\begin{figure*}
\begin{floatrow}
\ffigbox{%
    \includegraphics[scale=.35]{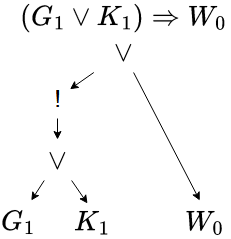}
}
{
  \captionof{figure}{\bf Parse Tree Example For Rule 4 of Table \ref{table::rule-examples}}%
  \label{fig:RuleParseTree}
}
\capbtabbox{%
    \null \small
    \centering
    \begin{tabular}{c|c} \hline
        \toprule
        \textbf {Original Clause} & \textbf{Translation} \\
        \midrule
        \tt IF A THEN B & \tt !A OR B \\
        \midrule
        \tt IF A THEN B ELSE C & \tt (!A OR B) AND (A OR C) \\
        \midrule
        \tt EXACTLY $k$ (...) & \shortstack{\tt AT MOST $k$ (...) AND \\ \tt AT LEAST $k$ (...) }\\
        \bottomrule
    \end{tabular}
    \vspace{-3.5mm}
}{
  \captionof{table}[t]{\bf Parsing Modifications}%
  \label{parse-modifications}
}
\end{floatrow}
\end{figure*}

\begin{table}%
\begin{minipage}[t]{\linewidth}{\null
\small
\centering
    \begin{tabular}{l|l}
        \toprule
        Term & Definition\\
        \midrule
        $\mathcal{S}$ & Set of feasible safety rules \\
        $\mathcal{D}$ & Set of all smart devices in the smart space \\
        $\mathcal{D}_s$ & Set of all possible device states \\
        $\mathcal{R}$ & Set of all routines in the smart space \\
        $h$ & Space state, a mapping $h: \mathcal{D} \rightarrow \mathcal{D}_s$ \\
        \bottomrule
    \end{tabular}
    \vspace{0.2cm}
    \captionof{table}[t]{\bf Key Terms and Variables}
    \label{tab:defs}
}
\end{minipage}
\end{table}

\begin{table}%
\begin{minipage}[t]{\linewidth}{\null
\small
\centering
    \begin{tabular}{l|l}
        \toprule
        {\bf Routine} & {\bf Explanation} \\
        \midrule
        \small
        \begin{tabular}{@{}l@{}}\texttt{outdoorCamera = ON;} \\ \texttt{frontDoor = LOCKED;} \\ \texttt{porchLight = ON;} \end{tabular} & \begin{tabular}{@{}l@{}}Evening Safety \\ \text{Setup} \end{tabular}\\
        \midrule
        \small
        \begin{tabular}{@{}l@{}}\texttt{thermo.mode = FROSTGUARD;} \\ \texttt{faucet.mode = DRIP;} \end{tabular} & \begin{tabular}{@{}l@{}} Cold Weather \\ Precautions \end{tabular} \\
        \bottomrule
    \end{tabular}
    \vspace{0.2cm}
    \captionof{table}[t]{\small\bf Examples of Routines %
    }
    \medskip
    \label{table::routine-examples}
}
\end{minipage}
\vspace{-0.5cm}
\end{table}

\section{Knox Algorithms}
\label{sec:Algorithms}

\newtheorem{Merging}[theorem]{Theorem}
\newtheorem{AttackDefense}[theorem]{Theorem}
\newtheorem{SafeStart}[theorem]{Theorem}
\newtheorem{IrreleventCutInUnsafeMerge}[theorem]{Theorem}
\newtheorem{MergingAccurate}[theorem]{Theorem}
\newtheorem{TriggersTheorem}[theorem]{Theorem}

\subsection{Approach}
\label{sec:algoapproach}

We outline Knox's three key contributions below and detail them in the following sections. Table \ref{tab:defs} shows variable definitions.

\noindent \textbf{Feasibility}:
{\it Given a set of safety rules, denoted as $\mathcal{S}$, is there at least one space state $h: \mathcal{D} \rightarrow \mathcal{D}_s$ that satisfies every rule in $\mathcal{S}$?} If rules are infeasible, no routine could ever be safe! Solving feasibility is equivalent to expressing safety rules as an SMT problem, so Knox uses an SMT solver, cvc5 \cite{DBLP:conf/tacas/BarbosaBBKLMMMN22}. The overhead is small since feasibility is only needs to be verified once for a given rule set, independent of routines, which change more frequently than rules. Hence we do not discuss Feasibility further.

\noindent \textbf{Single Routine Safety}: %
{\it Does a single routine violate any given rule in $\mathcal{S}$?} Table \ref{table::routine-examples} provides routine examples, from which we introduce the {\it cut}, a point between execution of consecutive routine commands. Cuts use rule semantics to make a per-rule worst-case assumption of any relevant unset states, avoiding an exhaustive search of all unset state combinations. We check safety at cuts and derive a causal relation that extends per-cut analyses.

\noindent \textbf{Concurrent Routine Safety}: {\it Are there any interleavings among a given set of routines $\mathcal{R}$ which would violate at least one safety rule in $\mathcal{S}$?} We assume $\mathcal{S,R}$ %
satisfy both Feasibility and Single Routine Safety.
We adapt the Single Routine Safety solution via {\it multi-routine cuts}: points in time between execution of consecutive commands that may be from different routines, while maintaining sequentiality within each routine.

\noindent \textbf{Completeness and Accuracy}: Knox is complete and detects any existing safety violations, but may be inaccurate, reporting false positives in certain conditions. Knox's safety assurances hold {\it no matter the initial space state}.

\subsection{NP-Hardness} 
\label{subsec:nphardnessproof}

\begin{restatable}[NP-Hardness: Single]{theorem}{NPHardSingle}
\label{thm::SATReduction}
The Single Routine Safety problem is NP-hard. \IfRestatedTF{}{{\small (Proof in Appendix)}}
\end{restatable}

The problem remains NP-hard even if the search starts from states that satisfy all safety rules, or \emph{safe states}.

\begin{restatable}[NP-Hardness: Safe Initial State]{theorem}{NPHardSafeInitialThm}
\label{thm::SATReduction2}
The Single Routine Safety problem, when limited to safe initial states, is NP-hard.\IfRestatedTF{}{ {\small (Proof in Appendix)}}
\end{restatable}

\begin{theorem}[NP-Hardness: Concurrent] The Concurrent Routine Safety problem is NP-hard.
\label{thm::SATReduction3}
\end{theorem}
\noindent \emph{Proof Sketch}: Since Single Routine Safety is a special case of Concurrent Routine Safety, this is also NP-hard. 
 
As perfect solutions to these problems are not tractable, Knox's approach makes use of heuristics and potentially allows false positives---a rule-satisfying \textit{safe} routine could be marked as rule-violating, or \textit{unsafe}.

\subsection{Stateful Trees} \label{subsec::StatefulTree}

Given one or more routines, analyzing a given tree (which represents one safety rule) requires capturing possible {\it state changes} caused by the routines. To do so, we instantiate each rule's parse tree $T$ as a \emph{stateful tree}:

\begin{definition}[Stateful Tree]
A stateful tree $ST$ is a map from each node in a rule parse tree $T$ to a truth value, either \emph{true} or \emph{false}: $ST: (n \in T) \rightarrow \{\textit{true}, \textit{false}\} $.
\end{definition}

Each node's state indicates whether that subclause (in the subtree rooted at that node) is considered satisfied (\textit{true}) or not (\textit{false}). Leaves are set to \textit{true} if the device state is correct, e.g., Rule 4 in Table~\ref{table::rule-examples} will have  {\it window} leaf ($W_0$ in Figure~\ref{fig:RuleParseTree}) set to \textit{true} if the window is indeed closed. For Rule 2, each of the three leaves is \textit{true} if its respective sprinkler is ON. Once leaves are set, we propagate values up to the root, following internal nodes' operations.

\begin{definition}[Rule Satisfaction/Violation]\label{def:satisfaction}
If a stateful tree $ST$ maps its root to \textit{true}, then the rule is satisfied, else violated.
\end{definition}

Determining Rule Satisfaction or Violation takes O($N_T$) time, where $N_T$ = number of nodes in $T$. %

\subsection{Representing Device State Through Cuts} \label{subsec::Cut}

Cuts (Section~\ref{sec:algoapproach}) capture the space state as it pertains to a single rule between consecutive commands, either from a single routine or interleaved across multiple routines which allows us examine safe and unsafe space states without searching all interleavings.
As we assume arbitrary initial device states, a device's state in a cut (a leaf in a rule parse tree) can take on one of many values:

\begin{definition}[Cut State] \label{def::cutState}
Each leaf $l$ in a rule parse tree is set to one of the following \textit{cut states} within a cut:
\squishlist
\item \textbf{Unset}: No routine has modified $l$'s truth value
\item \textbf{Conflict}: Two or more routines have set \textit{differing} values
\item \textbf{True}: A routine set $l$ to \emph{true}, and no others set $l$ to \emph{false}
\item \textbf{False}: A routine set $l$ to \emph{false}, and no others set $l$ to \emph{true}
\squishend
\end{definition}

Now we can formally define:

\begin{definition}[Single Routine Cut]
\label{def:SingleRoutineCut}
Given a rule parse tree $T$ and routine $R_i$, a \emph{single routine cut} contains:
\squishlist
\item \textbf{Location}: Similar to a ``Program Counter'', an integer $i \geq 0, \leq |R_i|$ representing how many commands of $R_i$ have completed executing, i.e., a point \emph{between} two commands.
\item \textbf{Cut State List}: A recording of the most recent \textit{cut state} for each leaf in $T$. Note that \emph{conflict} is not possible for a single routine.
\squishend
\end{definition}

\noindent {\underline{Example}}: Consider Table~\ref{table::rule-examples}'s Rule 1:
\noindent
\texttt{thermostat > 60 AND thermostat < 80}, with Routine \texttt{thermostat = 72}. The two leaves ($>60$ and $<80$) are both \textit{unset} %
at location 0; at location 1, a temperature of $72$ satisfies both conditions, so both map to \textit{true}.
Next we handle multiple interleaving routines:

\begin{table}
{\small
    \centering
    \caption{\bf Merge Cut State Operation}
    \label{table::cut-state-merge}
    \begin{tabular}{cc|c}
        \toprule
        Operand 1 & Operand 2 & Result \\
        \midrule
        Any cut state $x$ & $x$ & $x$ \\
        Any cut state $x$ & \textit{unset} & $x$ \\
        Any cut state $x$ & \textit{conflict} & \textit{conflict} \\
        \textit{true} & \textit{false} & \textit{conflict} \\
        \bottomrule
    \end{tabular}
    }
    \vspace{-3.5mm}
\end{table}

\begin{definition}[Multi-Routine Cut]
Given a rule parse tree $T$ and set of single routine cuts $C$ across routine set $R_C$, the \emph{multi-routine cut} is the result of \textit{merging} all cuts in $C$ to form a single merged cut. It contains: 
\squishlist
\item \textbf{Location Vector}: A vector of the {\it locations} $l_i$ within each routine $R_i$ in $R_C$: $<l_1,...,l_i,...,l_{|R_C|}>$.
\item \textbf{Merged Cut State List}: Most recent \textit{merged} \textit{cut states} for each leaf in $T$. Merging is done via the rules in Table \ref{table::cut-state-merge} and is both commutative and associative. 
\squishend
\end{definition}

\noindent {\underline{Example:}} Consider Rule 1 again, with two concurrent routines,
Routine 1: \texttt{thermostat = 72}, and
Routine 2: \texttt{thermostat = 68}, and the merged cut at location $1$ of each. The ultimate thermostat setting may be unknown (since  routines are concurrent), but this is unnecessary for safety checking---both single routine cuts set their leaves ($>60$ and $<80$) to \textit{true}, so the merged state is also true. Had Routine 2 instead been \texttt{thermostat = 90} (violating the rule), the merged state would be a conflict.

\noindent \emph{Ordering Cuts}:
\label{sec:causalcuts}
We define a {\it causal} relationship among cuts, used in later analysis. It captures whether one cut follows another and applies transitively: a cut $C$ is a causal predecessor of $C'$ if, starting from $C$, one or more commands are executed to reach $C'$. Formally:

\begin{definition}[Predecessor \& Successor Cuts] \label{def:pred}
Let  cuts $C, C'$ be defined as $C = (I, M), C' = (I', M')$, where $I$ and $I'$ are  location vectors $\{l_1, ... , l_n\}$ and $\{l'_1, ..., l'_n\}$, respectively, and $M$ and $M'$ are the respective merged cut state lists. Then $C'$ {\underline{precedes}} $C$ iff for every $i \in \{1...n\}$, $l'_i \leq l_i$ and $\exists j \in \{1...n\}$ such that $l'_j < l_j$. We also say that $C$ {\underline{succeeds}} $C'$.
\end{definition}

\begin{figure}
    \centering
    \includegraphics[width=\columnwidth]{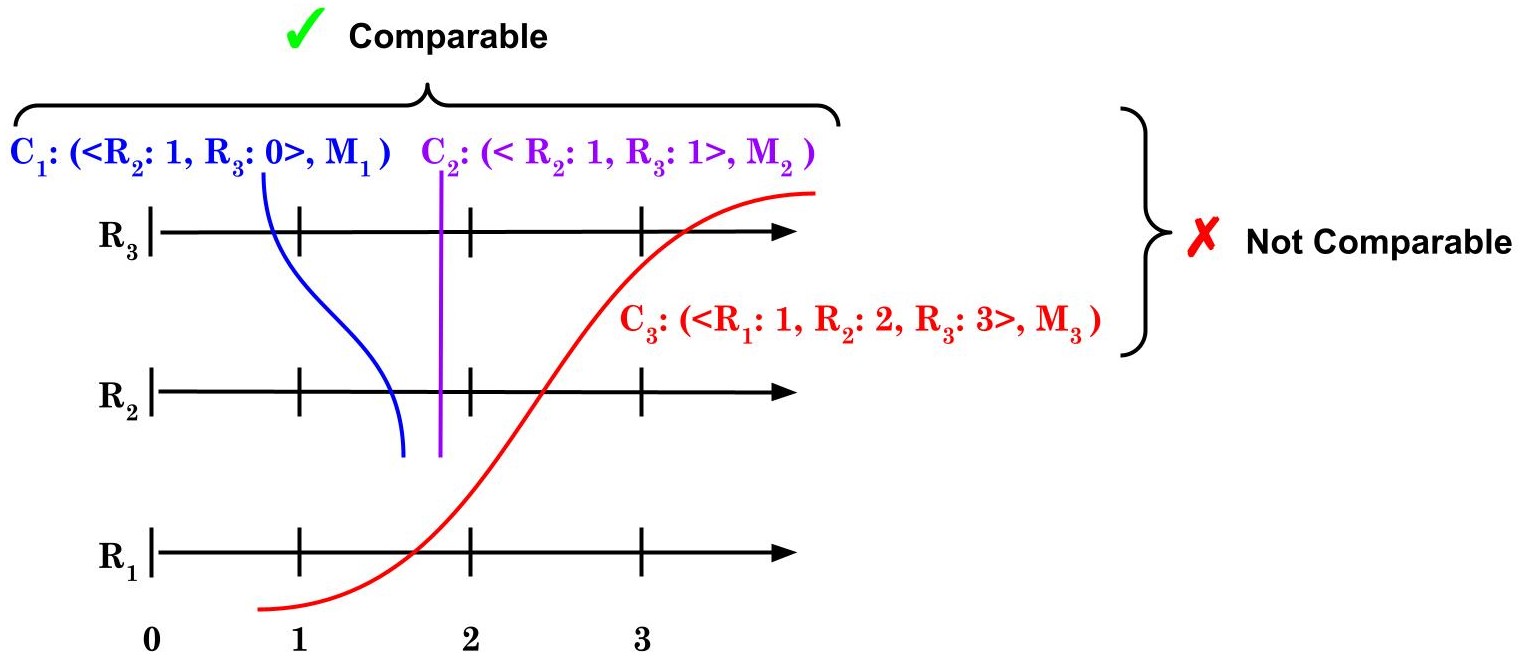}
    \caption{\small\bf Cuts and Cut Comparisons: \textnormal{
        Each routine can execute commands independently of others. Cuts can be compared if they share the same routine set --- $C_1$ and $C_2$ cover $R_1$ and $R_2$ so they are comparable, but $C_3$ includes state from $R_3$, so it cannot be compared to either.
    }} %
    \medskip
    \small

    \label{fig:CutTimelines}
\end{figure}

\begin{definition}[Immediate Predecessor and Immediate Successor Cuts] \label{def:immPred}
Let  cut $C$ have predecessor $P$. $P$ is an immediate predecessor of $C$ if and only if $l_i' = l_i$ for every $i \neq j \in \{1...n\}$ and $l'_j = l_j - 1$ for some $j$. Hence, we also call $C$ an immediate successor of $P$.
\end{definition}

\noindent This notion has similarities to vector timestamps: Figure~\ref{fig:CutTimelines} shows $C_1$ preceding $C_2$, yet both are incomparable to $C_3$.

\subsection{Node Contribution} \label{subsec::Contribution}

When a leaf's cut state is \emph{unset} or \emph{conflict} (Section~\ref{subsec::Cut}), its ultimate truth value---\textit{true} or \textit{false}---is still to be determined. Rather than enumerating all combinations across leaves, Knox adopts a heuristic that errs on the side of safety. We choose each leaf's {\it worst-case} truth value, i.e., the one that brings the stateful representation of the root ``closer'' to false.

To calculate this value for each node in a  tree, we introduce the notion of that node's \textit{contribution}. %
A node's setting \textit{contributes} to exactly one of rule satisfaction or violation. A child set to \textit{true} contributes to satisfaction of AND, OR, and AT LEAST subtrees. A child set to \textit{false} contributes to satisfaction of NOT (`!' in the grammar) and AT MOST subtrees.

For instance, in a disjunction $A \lor B \lor C$, setting more terms (leaves) to false brings the overall statement closer to false.
As rules grow larger or use more complex {Safron} grammar elements (Section~\ref{sec:grammar}), a leaf setting's effect on overall safety is no longer clear at a glance. Hence we define, for root and descendant:

\begin{definition}[Root Contribution]
\label{defn::RootContribution}
Let $r$ be the root of a tree $T$ with stateful tree mapping $ST$. If $ST(r) = \textit{true}$, then we say $r$ \textit{contributes} to rule satisfaction. If $ST(r) = \textit{false}$, then we say $r$ \textit{contributes} to rule violation.
\end{definition}

\begin{definition}[Descendant Contribution]
\label{defn::Contribution}
Let $A, D_1, ...D_N$ be nodes of a tree $T$ with stateful tree mapping $ST$, and let $\{D_1...D_n\}$ be descendants of node $A$. Let $A$ have initial truth value $s_A$. For arbitrary $i \in \{1...n\}$, there is a combination of truth values $\{s_{D_1}, s_{D_2}, ...s_{D_n}\} - \{s_{D_i}\}$  such that when $D_i$ is set to value $s_{D_i}$, $A$'s value is changed to $s_A'$, defined as $s_A' = \neg s_A$.
$s_A'$ contributing to rule satisfaction for the subtree rooted at $A$ is necessary and sufficient to state that $s_{D_i}$ contributes to rule satisfaction for $A$. The same holds for rule violation.

\end{definition}

For contribution to be usable in analyzing state changes on a tree, it must be deterministic, meaning that a state assignment should have \emph{one} outcome with respect to rule satisfaction or violation. The following theorem addresses this concern:

\begin{restatable}[One-To-One Contribution]{theorem}{onetoonecontribution}
\label{thm::Contribution}
Given a tree $T$ for a rule, and a node $n \in ST$, a specific truth value of node $n$ can contribute only to {\it exactly one} of either rule satisfaction or rule violation.
\IfRestatedTF{}{{\small (Proof in Appendix)}} 
\end{restatable}

For example, in Table~\ref{table::rule-examples}'s Rule 4
(tree in Figure~\ref{fig:RuleParseTree}), turning off the backyard grill contributes to rule satisfaction. If the smoker is also off, the parent OR is set to false, inductively bringing the root closer to true. The same action can \emph{never} make the parent OR true, so it \emph{only} contributes to rule satisfaction.

\subsection{Attack and Defense States}
\label{sec:attackdefense}

Given a safety rule (tree), a node state that brings the system closer to violating it is called an {\it attack}  state. One that takes it away from violation is called a {\it defense} state. Formally:

\begin{definition}[Attack and Defense States]\label{att/defNameDef}
Given a rule parse tree $T$ and a stateful tree $ST$, if a node $n \in T$'s truth value $ST(n)$ \textit{contributes} to $ST$'s violation, we say $n$ is in its \textit{attack state} $A$.
In contrast, if a node $n \in T$'s value $ST(n)$ \textit{contributes} to $ST$'s satisfaction, we say $n$ is in its \textit{defense state} $D$.
\end{definition}

By Theorem \ref{thm::Contribution}, each possible truth value of a node is either an attack or a defense state. Hence every node \textit{must} have both, depending on the tree structure.

Now we can present Algorithm~\ref{alg:determineAtkDef}  to calculate  attack and defense states. %
Analyzing a rule parse tree $T$ takes $O(N_T)$ time, where $N_T$ is the number of nodes in $T$. For all rules, this takes $O(N_\mathcal{S})$ time, where $N_\mathcal{S} = \sum_{T\in\mathcal{S}} N_T$ is the total number of nodes across all trees.

\begin{algorithm}
\caption{\bf Determining attack/defense states}
\label{alg:determineAtkDef}
{\small
\KwIn{Rule parse tree $T$} %
\KwOut{Modified rule parse tree with att/def states on nodes}
\ForEach{node $n$ in a pre-order traversal of $T$}{
    \uIf{$n$ is the root}{
        \{$n$.attackState, $n$.defenseState\} $\gets$ \{\textit{false}, \textit{true}\}\;
    }
    \uElseIf{$n$'s parent's operation is AND, OR, or AT\_LEAST}{
        \{$n$.attackState, $n$.defenseState\} $\gets$
        \{$n$.parent.attackState, $n$.parent.defenseState\}
    }
    \uElse{
        \{$n$.attackState, $n$.defenseState\} $\gets$
        \{$n$.parent.defenseState, $n$.parent.attackState\}\;
    }
}
}
\end{algorithm}

\noindent \textbf{Choosing Unset/Conflict States}: When a leaf's \textit{cut state} (Section~\ref{subsec::Cut}) is either \textit{unset} or \textit{conflict}, we err on the side of caution, setting each of these leaves to their {\it attack} states (contributing to rule violation). %
A corner case occurs when multiple leaves within a rule refer to the same device, which may produce contradictory states.
Ex.: for the degenerate rule \texttt{AC == ON OR AC == OFF} with unknown AC state, we assume the worst case for \textit{both} leaves: leaf \texttt{AC == ON} assumes the AC is off, and leaf \texttt{AC == OFF} assumes it is on.
This is naturally one reason why Knox may experience false positives, though recurrence of a device in multiple clauses is rare in simple safety rules. %
(This degenerate rule is vacuously true, so it would not be admitted to the system.) %

\subsection{Determining Cut Safety}
\label{subsec:cut-context}

We now leverage cuts (Section~\ref{subsec::Cut}) to check for rule violations without having to enumerate all possible interleavings. The \textit{Cut Contextualizer} algorithm (Algorithm~\ref{alg:cut-context}) starts from a cut and constructs a stateful tree. It then makes worst-case decisions for leaves set to \textit{unset} or \textit{conflict}, while retaining settings of \textit{true} or \textit{false}.

\begin{algorithm}
\caption{\bf Cut Contextualizer}
\label{alg:cut-context}
{\small
\KwIn{Rule parse tree $T$, Cut $C$}
\KwOut{Stateful tree $ST$}
\ForEach{leaf $l$ $\in$ $T$}{
    $ST(l)$ $\gets$ $C(l)$ if $C(l)$ is \textit{true} or \textit{false}\;
    $ST(l)$ $\gets$ $l$.attackState otherwise\;
}
Evaluate $ST$ (Section~\ref{subsec::StatefulTree})\;
return $ST$\;
}
\end{algorithm}

\begin{theorem}\label{def:cut-safe}
A cut $C$ is \textit{safe} at its location vector if the stateful tree $ST$ via the \textit{Cut Contextualizer} is satisfied.
\end{theorem}

\begin{proof}
The stateful tree $ST$ is an explicit evaluation of a safety rule (Section~\ref{sec:grammar}). Leaves set to \textit{true} or \textit{false} by the cut are guaranteed---all interleavings agree on this state. Leaves set to \textit{conflict} or \textit{unset} assume their attack state, which by Theorem~\ref{thm::Contribution} \textit{must} contribute to rule violation. Thus if the stateful tree's root is set to \textit{true}, then the values of explicitly known devices are sufficient for rule satisfaction.
\vspace{-1.8mm}
\end{proof}

In Algorithm~\ref{alg:cut-context}, cut safety in a rule parse tree $T$ with $N_T$ nodes takes $O(N_T)$ time. A set of rules in $\mathcal{S}$ takes $O(N_{\mathcal{S}})$ time, where $N_{\mathcal{S}}=\sum_{T \in \mathcal{S}}{N_T}$. %

\subsection{Correctness of Routine Safety} \label{subsec:RoutineSafetyCorrectness}

To avoid exploring all interleavings among routines while checking for safety, we combine the fact that cuts record the worst-case result of interleavings up to their location vector with the notion of cut causality from Section~\ref{sec:causalcuts}. The result is that cut safety can be extended to ensure interleaving safety, which we formalize below as a correctness theorem:

\begin{restatable}[Correctness of Routine Safety]{theorem}{correctness}
\label{thm::Merging}
Assume we have two safe multi-routine cuts $C1$ and $C2$, such that $C1$ precedes $C2$. If all cuts that succeed $C1$ and precede $C2$ are safe, then there exists no unsafe interleaving of the commands between $C1$ and $C2$'s location vectors. 
\IfRestatedTF{}{{\small (Proof in Appendix)}}
\end{restatable}

\begin{corollary}\label{cor::knox_correct}
If all cuts in a routine or routine group are safe, then single/concurrent routine safety is satisfied.
\end{corollary}

\noindent The above provides explainability for rule violation: 
\squishenumerate
    \item Given an unsafe interleaving, we can point to exactly which safety rule was violated.
    \item We can provide the exact commands that may cause an unsafe interleaving for a rule by taking the most recent settings of non-defense leaves.
    \item Using the causal relationship, we can determine that interleavings between subgroups of the checked routine set are safe. This holds even if an unsafe interleaving exists at a different location vector.
\squishenumerateend

\subsection{Knox's Baseline Routine Safety}\label{subsec::Mlirad}

Algorithm~\ref{alg::core} provides an overview of our system. We first analyze rules for their attack/defense states based on contribution (Section \ref{subsec::Contribution}). Then, Algorithm \ref{alg::single} finds and reports statically safe and unsafe individual routines by constructing and analyzing a stateful tree for each rule and single-routine cut. Finally, Algorithm \ref{alg::concurrent} finds and reports all unsafe cuts by constructing and checking multi-routine cuts for every location vector in the routine set.

\begin{algorithm}
\caption{\bf Core System}
\label{alg::core}
{\small
\KwIn{Feasible safety rule parse trees $\mathcal{S}$, Routines $\mathcal{R}$, Devices $\mathcal{D}$, Device States $\mathcal{D}_s$}
\KwOut{SafeRoutines, UnsafeRoutines, UnsafeCuts}
\ForEach{$T \in \mathcal{S}$}{
    Analyze leaves in $T$ for contribution (Section~\ref{subsec::Contribution})\;
}
Construct all single routine cuts\;
SafeRoutines, UnsafeRoutines $\gets SingleRoutineSafety(\mathcal{S}, \mathcal{R}, \mathcal{D}, \mathcal{D}_s)$ (Algorithm~\ref{alg::single})\;
UnsafeCuts $\gets ConcurrentSafety(\mathcal{S}, \text{SafeRoutines}, \mathcal{D}, \mathcal{D}_s)$ (Algorithm~\ref{alg::concurrent})\;
return \{SafeRoutines, UnsafeRoutines, UnsafeCuts\}\;
}
\end{algorithm}

\begin{algorithm}
\caption{\bf Single Routine Safety (Baseline Knox)}
\label{alg::single}
{\small
\KwIn{Feasible safety rule parse trees $\mathcal{S}$, Routines $\mathcal{R}$, Devices $\mathcal{D}$, Device States $\mathcal{D}_s$}
\KwOut{SafeRoutines, UnsafeRoutines}
SafeRoutines, UnsafeRoutines $\gets \varnothing$\;
\ForEach{$R \in \mathcal{R}$}{
    \ForEach{$T \in \mathcal{S}$}{
        \uIf{$R$ has no effect on $T$}{
            continue\;
        }
        \ForEach{Single Routine Cut $C$ in $R$}{
            Stateful Tree ST $\gets CutContextualizer(C)$\;
            \uIf{ST violated}{
                UnsafeRoutines $\gets$ UnsafeRoutines $\cup \{R\}$\;
                go to next routine\;
            }
        }
    }
    SafeRoutines $\gets$ SafeRoutines $\cup \{R\}$\;
}
return SafeRoutines, UnsafeRoutines\;
}
\end{algorithm}

\begin{algorithm}
\caption{\bf Concurrent Routine Safety (Baseline Knox)}
\label{alg::concurrent}
{\small
\KwIn{Feasible safety rule parse trees $\mathcal{S}$, Routines $r_{1..n}$, Devices $\mathcal{D}$, Device States $\mathcal{D}_s$}
\KwOut{UnsafeCuts}
UnsafeCuts $\gets \varnothing$\;
\ForEach{location vector $<i_1,...,i_{|\mathcal{R}|}>$}{
    Construct multi-routine cut $C$ for this location vector\;
    Stateful Tree ST $\gets CutContextualizer(C)$\;
    \uIf{ST violated}{
        UnsafeCuts $\gets$ UnsafeCuts $\cup \{ C \}$\;
    }
}
return UnsafeCuts\;
}
\end{algorithm}

\subsubsection{Runtime Complexity Analysis}
\label{CoreSystemAnalysis}

\begin{table}
{\centering\small
    \vspace{2mm}
    \begin{tabular}{c|l}
        \toprule
        Variable & Definition\\
        \midrule
        $R_{max}$ & The maximum length routine \\
        $K$ & Total single routine cut locations
        across all routines \\
        $V$ & Total multi routine cut locations across all routines \\
        $N_T$ & Node count in tree $T$ \\
        $N_\mathcal{S}$ & Node count across all trees of safety rules in $\mathcal{S}$ \\
        $L_T$ & Leaf count in tree $T$ \\
        $L_\mathcal{S}$ & Leaf count across all trees of safety rules in $\mathcal{S}$ \\
        \bottomrule
    \end{tabular}
    \vspace{-2mm}
    \caption{\bf Recurring Runtime Analysis Variables} 
    \label{tab:analysisvars}
    \vspace{-2mm}
    }
\end{table}

\textbf{Single Routine Safety}: To construct all single routine cuts,
we consider all locations. We define the total number of locations across all routines as $K = \sum_{R_i\in\mathcal{R}} (|R_i| + 1)$. We assume a data structure mapping devices to leaves in rule parse trees, built when each routine is inserted in time proportional to the routine's size. When examining a command, we consult the mapping to find the leaves modified while constructing a cut. For each location and rule, we construct the leaf mapping in $O(L_T)$ time, where $L_T$ is the number of leaves in tree $T$. Defining $L_\mathcal{S} = \sum_{T\in\mathcal{S}} L_T$, constructing all single routine cuts across all rules takes $O(K \cdot L_\mathcal{S})$ time.

For single routine safety, we simply check each of these single routine cuts for safety using the Cut Contextualizer (Algorithm~\ref{alg:cut-context}), which takes $O(K\cdot N_\mathcal{S})$ time.

\noindent \textbf{Concurrent Routine Safety}: This checks multi-routine cuts at each possible location vector. Define  $V=\prod_{R_i\in\mathcal{R}} (|R_i| + 1)$ as  total number of location vectors. To construct all multi-routine cuts, we follow a lattice (using memoization), merging in one new single routine cut in $O(L_\mathcal{S})$ time each iteration for a total runtime of $O(V\cdot L_\mathcal{S})$.
For each, we must run the Cut Contextualizer to determine safety, for a total runtime of $O(V\cdot (L_\mathcal{S} + N_\mathcal{S}))$. We note that $V=O(|R_{max}|^{|\mathcal{R}|})$, where $R_{max}$ has the maximum number of commands of any routine.

The Cut Contextualizer complexity is {\it linear}, as opposed to using SMT to solve the Single and Concurrent Routine Safety problems, which would have made the complexity exponential. %

\noindent \textbf{Adding and Removing Routines and Safety Rules}: When a new rule is inserted, %
the above algorithms only check existing  routines against that rule.
A new routine is checked only against all existing rules, along with cuts involving all routines (this is unavoidable as the new routine will have interleavings with existing ones). %
Removal of routines or safety rules never makes a safe system unsafe, though cuts previously considered unsafe may be updated. %

\subsection{Reducing False Positives}
\label{reducingFalsePositives}

By design, Knox misses no unsafe cuts, but as it errs on the side of safety it may be prone to false positives: points in execution where a cut is identified as unsafe, but no interleaving from some initial state violates a rule. Under certain conditions, however, Knox can guarantee these cannot occur:

\begin{restatable}[Single-Routine: Avoidance of False Positives]{theorem}{singleroutinefp}
\label{SingleRoutineFalsePositives}
For a single routine, and a rule parse tree $T$, if $T$ has at most one leaf for each device in the smart space, then Knox avoids false positives.
\IfRestatedTF{}{({\small Proof in Appendix})}
\end{restatable}

\noindent We first extend notions of attack and defense states  (from Section~\ref{sec:attackdefense}):

\begin{definition}[Command Classification:
Attack, Defense, and Ambiguous Commands] \label{def::commandclass}
For a rule parse tree, a command $c$ is: 
(i) an \textit{\underline{attack command}} if $c$ sets all affected leaves to attack states; (ii) a \textit{\underline{defense command}} if $c$ sets all affected leaves to defense states; (iii) otherwise it is an \textit{ambiguous command}.
\end{definition}

\noindent We call this derivation  \emph{Command Classification}.
Intuitively, attack commands can only ever contribute to rule violation while defense commands do the opposite (w.r.t.  a given rule). Ambiguous commands need special care, %
as they induce both attack and defense states.

\noindent {\underline{Example}}: Consider Table \ref{table::rule-examples}'s Rule 3. Command AC = OFF is a defense command: the leaf AC == OFF would be set to \textit{true}, which is its defense state. Command window = OPEN is an attack: setting leaf window == OPEN to \textit{true} would put the leaf in attack state. 

\noindent {\underline{Example}}: Consider the rule \texttt{($A_1$ AND $B_2$) OR $A_2$}. The command $A = A_1$ is ambiguous, setting leaf $A == A_1$ to its defense state \textit{true}, but $A == A_2$ to its attack state \textit{false}.

\noindent Now we can prove that under some conditions, multi-routine safety incurs no false positives:  

\begin{restatable}[Multi-Routine: Avoidance of False Positives]{theorem}{multiroutinefp}
An unsafe multi-routine cut $C$ exists if and only if an unsafe interleaving exists for some initial state, provided either conditions (A1 AND A2) are true, or B is true, i.e.:
\squishlist
\item (A1) The rule parse tree $T$ has at most one leaf for each device in the smart space, AND
\item (A2) For all leaves in $C$ set to \textit{conflict}, all routines execute attack commands \textit{after} defense commands
\squishend
OR
\squishlist
\item (B) At most one leaf in $C$ is set to \textit{conflict}. 
\squishend
\end{restatable}

\section{Optimizations} \label{sec:opt}

Two optimizations further reduce safety checking runtime---\emph{Attack/Defense} and \emph{Wall/Siege}.

\subsection{Attack/Defense} \label{subsec::AttDef}
First, we re-apply the Command Classifications on attack and defense commands (Definition~\ref{def::commandclass}). This helps skip certain safety checks without violating completeness.

\begin{restatable}[Irrelevant Cuts]{theorem}{AtkDefThm}
\label{thm::AttackDefense}
A cut that satisfies one of:
\squishlist
\item {\bf \textit{Case 1}}. Has a \textbf{safe} immediate predecessor, where $l'_j$ is a \textbf{defense} command (recall Definition~\ref{def:immPred})
\item {\bf \textit{Case 2}}. Has an immediate successor (no safety qualification needed), where $l'_j$ is an \textbf{attack} command
\squishend
are \textbf{irrelevant}, and are not needed to detect rule violation.
\end{restatable}

\begin{proof}
Consider an arbitrary cut $C$. We consider both cases: %

\noindent \textbf{Case 1}:
Our cut has some \textbf{safe} immediate predecessor, with  the differing command a \textbf{defense}.
A defense  only allows the leaf it touches to contribute towards rule satisfaction:
\squishlist
\item If the leaf was in its attack state in the immediate predecessor, then it will be a defense state if the attack was from the same routine or a \textit{conflict} if not.
\item If the leaf was in \textit{conflict} or defense states, it remains. %
\item If the leaf was \textit{unset}, it will be set to its defense state.
\squishend
In all three subcases, no leaf that contributed to rule satisfaction is changed to contribute to violation. Since the predecessor cut was already safe and the only changes increase the total leaves contributing to rule satisfaction, $C$ remains safe, with no further checking required.

\noindent \textbf{Case 2}:
Our cut has some immediate successor, where the differing command is an \textbf{attack}.
Attack commands allow only contributions to rule violation:
\squishlist
\item If the leaf was in its defense state in $C$ then it is set to its attack state in the immediate successor (if defense was from the same routine as the attack command) or a \textit{conflict} if not.
\item If the leaf was \textit{unset}, then it is set to its attack state in the immediate successor.
\item If the leaf was in \textit{conflict} or attack states, it remains. %
\squishend
In all 3 subcases, the immediate successor's leaf records an attack state. So if $C$ was unsafe, all changes only raise the number of leaves violating the rule. Hence the immediate successor is unsafe, and we can skip checking $C$.
\end{proof}

\begin{restatable}{theorem}{SafeStart}
\label{thm::SafeStart}
If the multi-routine cut located before any routine commands execute is safe, then any cut located immediately {\it after} a defense command is irrelevant and can be skipped for safety checking. 
\IfRestatedTF{}{{\small (Proof in Appendix)}}
\end{restatable}

Theorems~\ref{thm::AttackDefense}~and~\ref{thm::SafeStart} are powerful because they imply that given a safe initial state, we need not evaluate immediate successors of defense commands.
Hence the Attack/Defense optimization modifies the Baseline Knox algorithm of Section \ref{subsec::Mlirad} by first executing Command Classification (Definition~\ref{def::commandclass}) and then assuming a safe start. Multi-routine cuts are verified only if they are located %
(i) after the end of an attack command sequence, or
(ii) after an ambiguous command.
Its runtime is:

\begin{restatable}[Attack/Defense Runtime]{theorem}{AttackDefenseRuntimeThm}
\label{AttackDefenseRuntimeThm}
With \({K_{ad} = \sum_{R_{i}\in\mathcal{R}} \frac{|R_i|}{2}}\),  runtime of checking concurrent routine safety after applying the Attack/Defense optimization is
\[
{\textstyle O(N_\mathcal{S} + K_{ad} \cdot L_\mathcal{S} + K_{ad} \cdot N_{\mathcal{S}} + \left(\frac{|R_{max}|}{2}\right)^{|\mathcal{R}|} \cdot (L_\mathcal{S} + N_\mathcal{S}))}
\]
{\small (Notations are in Table~\ref{tab:analysisvars})}. 
\IfRestatedTF{}{ {\small (Proof in Appendix)}}
\end{restatable}

\subsection{Wall/Siege} \label{subsec::WallSiege}

The nature of cut state \textit{merging} (Table~\ref{table::cut-state-merge}) leads us to another optimization. As we progress through the concurrent routine safety algorithm, we can limit our focus to only the leaves which contribute to rule satisfaction. The merge operation guarantees that this set can only shrink as we check more routines, allowing a quick short-circuit once it is empty (i.e., once nothing can satisfy the rule). As the set shrinks, it also determines a minimal set of leaves sufficient for rule satisfaction, letting us safely ignore changes to leaves outside it. We thus become \textit{more} efficient as checking progresses, with the potential to outperform the Attack/Defense optimization. \emph{Combining} these short-circuit and shrinking-set optimizations powerfully reduces the number of checks required by the base algorithm.

Formally, once a leaf in a cut is set to \textit{conflict} or its attack state (Definition~\ref{def::cutState}), it will {\it continue to remain} in its attack state in the %
search, 
regardless of any future merges~\footnote{Recall the merge operation from Table \ref{table::cut-state-merge}: if an attack state is merged with a matching state, it is preserved; else the merge causes a Conflict, which  Cut Contextualizer (Section \ref{subsec:cut-context}) presumes is attack.}. Call the collection of leaves that have not entered this preserved attack state a {\it \underline{Wall}}. We then define safety-conforming procedures called {\it \underline{Sieges}} that remove leaves from it: %

\begin{definition}[Wall]\label{def::wallDefinition}
For a given rule parse tree $T$, its leaves $T_L$, and a cut $C$, a wall is a partial map: \\$W : T_L \to \{\textit{true}, \textit{false}, \textit{unset}, \textit{conflict}\}$ such that any leaves $l$ that are Unset or in their Defense State are contained within the partial mapping. These leaves $l$ have a defined $W(l) = M(l)$ giving their known cut state with respect to $C$. We say this wall corresponds to cut $C$.
\end{definition}

\begin{definition}[Siege]\label{def::seige}
    A single routine cut $C$ can {\it siege} the wall $W$ and produce a new wall $W'$ by: 
        (1) Removing any leaf mappings in $W$ that record an attack or conflict cut state in $C$, 
        (2) Updating state of any leaves in $W$ that previously recorded an Unset state to a defense state recorded in $C$. 
    Let $W$ correspond to cut $D$. Then we say $W'$ corresponds to merged result of $C$ and $D$.
\end{definition}
Note that a siege can only preserve or reduce a wall's total leaf mappings, i.e., walls change only monotonically after construction.  Sieges function similarly to merges (Section \ref{subsec::Cut}) but occur between a wall and a single routine cut rather than two cuts of any disjoint single routine composition.

\begin{restatable}[Wall Equivalence]{theorem}{WallEqThm}
\label{thm::siege1.2}
For an arbitrary wall $W$ corresponding to an arbitrary cut $C$, let $L_1$ be a set of truth values defined by both criteria:
    (1) All leaves inside $W$ are set to a value defined by the Cut Contextualizer, 
    (2) All leaves not inside $W$ are set to their attack state.
Let $L_2$ be a set of truth values generated by applying the Cut Contextualizer to $C$. Then $L_1$ = $L_2$.
\IfRestatedTF{}{{\small (Proof in Appendix)}}
\end{restatable}

If a cut $C$ does not remove any leaves from a wall after a siege, then it will similarly be ineffective on sieging a smaller wall (one with a subset of the original's leaf mappings). So:

\begin{restatable}[Wall Transitivity]{theorem}{WallTransitivityThm}
\label{thm::siege3}
Consider a single routine cut $C$ and a \textbf{safe} multi-routine cut $M$ with a 0 component for $C$'s routine. Let $M$'s wall be $W_M$, and assume sieging with $C$ removes no leaves. Then any other \textbf{safe} multi-routine cut $M'$ sharing $M$'s non-zero location vector components (0 component for $C$'s routine) is safe after merging with $C$. 
\IfRestatedTF{}{{\small (Proof in Appendix)}}
\end{restatable}

\begin{restatable}[Wall-Siege Concurrent Routine Safety Runtime]{theorem}{ConSafRuntimeWallThm}
\label{ConSafRuntimeWallThm}
Using walls and sieges, concurrent routine safety runtime is: \(
{\textstyle O(V\cdot L_\mathcal{S} + \binom{|\mathcal{R}|}{L_\mathcal{S}} \cdot |R_{max}|^{L_\mathcal{S}} \cdot N_\mathcal{S})}
\) 
\IfRestatedTF{}{{\small (Proof in Appendix)}}
\end{restatable}

This  is smaller than  Baseline's runtime from Section~\ref{CoreSystemAnalysis}. %

\begin{figure}[t]
    \centering\small
    \includegraphics[scale=.42]{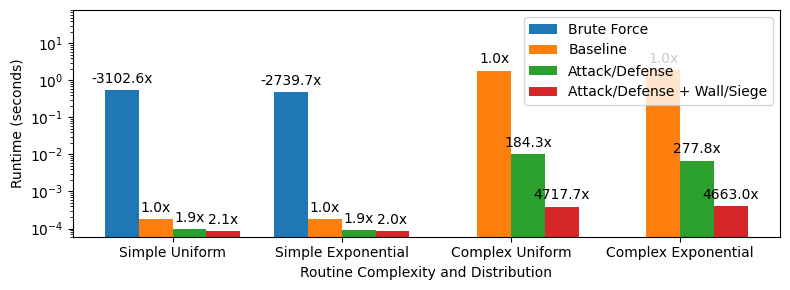}
    \vspace{-.5cm}
    \medskip
    \captionof{figure}{\bf Average runtime comparison for routines with different complexity and distribution.}%
    \label{fig:runtime}
\end{figure}

\section{Evaluation} \label{sec:eval}

We implemented Knox in C++ and integrated it with the popular HomeAssistant framework~\cite{HomeAssistant}. We 
answer four research questions:

\begin{enumerate}[nosep]
\item How fast is Knox at checking a routine+safety clause set? 

\item How does Knox scale with number of %
safety clauses?

\item What is Knox's false positive/negative behavior? %

\item How does Knox perform against TapChecker~\cite{TapChecker}?
\end{enumerate}

Due to the scarcity of real smart space benchmarks, we evaluate Knox under a diversity of scenarios, using synthetic, realistic, and trace-based workloads.

\subsection{Baseline Comparisons}
\vspace{-1mm}

We simulate a medium-sized home with 20 smart devices, 1 hub, and safety clauses touching 3-7  devices. (These scales are typical of today's homes.) We generate routines and safety rules with varying concurrency, length, and distribution. 

Figure~\ref{fig:runtime}  compares  
Baseline Knox (Section~\ref{subsec::Mlirad}),
Knox + Attack/Defense Optimization (Section~\ref{subsec::AttDef}), vs. Knox + Attack/Defense + Wall/Siege (Section~\ref{subsec::WallSiege}). We show the runtime to check the entire routine+safety clause set for all safety violations,  under varying routine workloads:  %
 (a) {\it Simple} workloads have 3 concurrent routines each with 3 commands, (b) {\it Complex} workloads have 7 concurrent routines each  with 7 commands. In each, routines pick devices based on either (i) exponential distribution (simulates more overlap among routines), or (ii) uniform distribution (less overlap). All data points use an average of 1000 trials.

We observe that Baseline Knox (1) significantly reduces runtime over the default brute-force approach,  %
and (2) achieves significant speedup, up to 3102$\times$ on simple routines (for complex routines, Brute Force took too long to measure). Further, (3) our two optimizations---Attack/Defense \& Wall/Siege---reduce runtime over Baseline Knox on complex routines, up to 277$\times$ and 4717$\times$ respectively. %

\begin{table}[t]
    \small \centering
    \vspace{0.2cm}
    \begin{tabular}{l|cccc}
    \toprule
    \textbf{Run Type} & $2{\times}2$ Unif & $2{\times}2$ Expo & $3{\times}3$ Unif & $3{\times}3$ Expo \\
    \midrule
    \textbf{Accuracy} & 0.989    & 0.988    & 0.961    & 0.963    \\
    \bottomrule
    \end{tabular}
    \vspace{-0.2cm}
    \caption{\small\bf Accuracy Rates for Various Routine Sets: \textnormal{
        Unsafe cut reports from Knox are compared against brute-force checks. False positive rates stay below 4\%.
    }} 
    \label{tab:tp_rates}
\end{table}

\subsection{Scalability}

\begin{figure}[t]
\centering \small
    \centering\small
    \includegraphics[scale=.47]{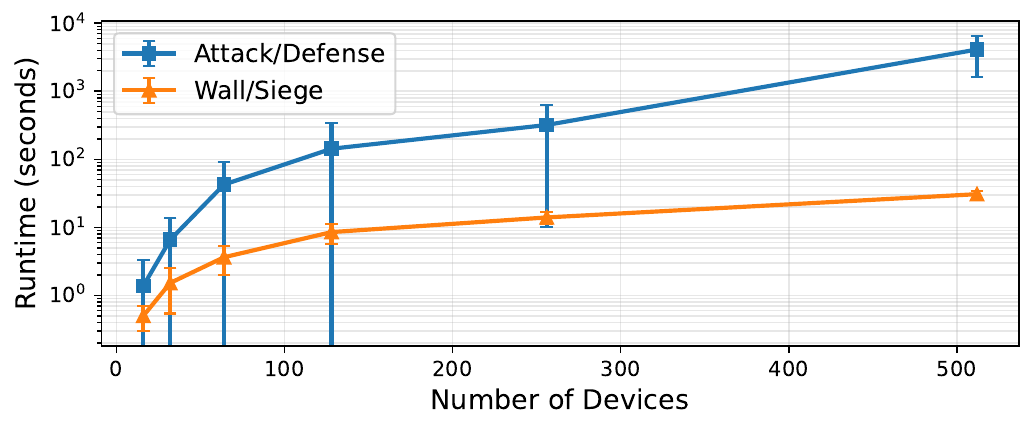}
    \vspace{-5mm}
    \medskip
    \captionof{figure}{\small \bf Device Scalability: 
    \textnormal{Average runtime vs. number of devices.}}%
    \label{fig:safetyrules}
\end{figure}

Table~\ref{tab:tp_rates} shows accuracy (Section~\ref{reducingFalsePositives}) with \{2,3\} concurrent routines $\times$ \{2,3\} commands and a randomly generated safety tree of 5 device states. Accuracy =  (1 - FPR), where the false positive rate (FPR) is the ratio of nonexistent unsafe cut reports from Knox to the total number of unsafe cut reports, measured by matching Knox's output against corresponding brute-force results. We observe (i) high accuracy and (ii) insignificant accuracy drop for complex routines, with no missed violations (false negatives = 0).

Figure \ref{fig:safetyrules} evaluates scalability with respect to devices.  %
We also scaled the size of the safety tree from 8 to 256 leaf nodes, (correlated with the  device count), along with routine sizes from $5\times5$ to $8\times8$. Each data point is averaged over 100 trials. Wall/Siege shows a consistent advantage over Attack/Defense with a $2.76\times$--$133.4\times$ speedup when scaling to hundreds of smart devices (Baseline Knox was prohibitively slow to test).

\subsection{Home Assistant Room Simulation}

\begin{figure}%
\centering \small
    \vspace{0.4cm}
    \centering\small
    \includegraphics[scale=.6]{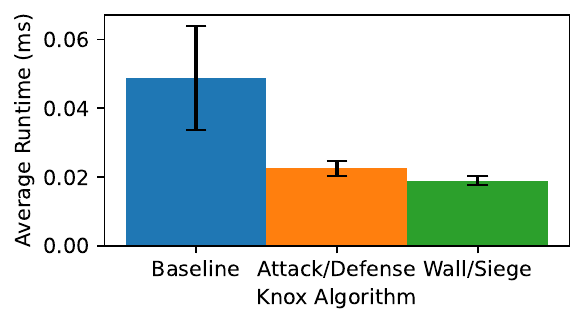}
    \vspace{-0.4cm}
    \captionof{figure}{\small \bf Home Assistant Experiments: \textnormal{Average runtime comparison for realistic routines in a room simulation.}}%
    \label{fig:roomsim}
\end{figure}

We use the HomeAssistant integration to generate a smart room %
containing 7 devices, 5 routines (2-5 commands/routine), and 2 safety rules with a total of 4 device states involving range queries %
($\lesseqgtr$ comparisons). %
Figure~\ref{fig:roomsim} shows that (i) 
Attack/Defense and Wall/Siege both outperform Baseline by around $2.2\times$, and (ii) Wall/Siege is slightly faster than Attack/Defense. (Brute Force was  prohibitively slow to test.)

\begin{table}[h]%
    \vspace{0.2cm}
    \centering\small
    \begin{tabular}{lllll}\toprule
        &\multicolumn{2}{c}{\textbf{50 Routines}}&\multicolumn{2}{c}{\textbf{500 Routines}}
        \\\cmidrule(r){2-3}\cmidrule(r){4-5}   
        &LR&BR&LR&BR\\\midrule
        TapChecker & .66 s & .65 s 
                & 69.365 s & 69.372 s
                \\
        Knox   & .029 s & .027 s 
                & .238 s & .239 s 
                \\
        \midrule
        Speedup & 22.76$\times$ & 24.07$\times$ 
                & 291.45$\times$ & 290.26$\times$ 
        \\\bottomrule
    \end{tabular}
    \vspace{-0.3cm}
    \caption{\small {\bf Knox and TapChecker~\cite{TapChecker} Performance and Routine Scalability}. \textnormal{Living Room (LR), Bedroom (BR) }}
    \label{tab:knoxtapunified}
\end{table} 

\vspace{-10pt}
\subsection{Comparison vs. TapChecker}

To compare Knox with TapChecker~\cite{TapChecker} in a fair way, we borrow that paper's workload of 2011 routines in total. Table~\ref{tab:knoxtapunified} evaluates scalability with respect to routines. It shows that Knox runtime is 22.76$\times$ to 291.45$\times$ faster than TapChecker. In terms of accuracy, TapChecker only considers pairwise (2-way) conflicts among routines. To compare fairly against Knox (which considers conflicts from any number of concurrent  routines), we {\it extended TapChecker's original implementation to detect \emph{3-way conflicts}}.  %
We use 17 routines, and then inject 12 crafted rules that can only be violated by 2-way or 3-way concurrent routines.  
Table \ref{tab:fn_rates}  %
shows that even the augmented TapChecker can have false negative rates (missed violations) up to 12.2\%, but Knox has none. %

\begin{table}[h]%
    \vspace{0.3cm}
    \centering\small
    \begin{tabular}{lccc}
        \toprule
        Three-way Rules Added & 1 & 3 & 4  \\
        \midrule
        TC-triple & 3.44\% & 8.89\% & 12.2\% \\
        Knox & 0\% & 0\% & 0\% \\
        \bottomrule
    \end{tabular}
    \vspace{-0.2cm}
    \caption{\small\bf False Negative (FN) Rates for Varying Numbers of \textit{Three-way Safety Rules}: \textnormal{Conflicts missed. %
    Knox never misses conflicts. 
    }}
    \label{tab:fn_rates}
\end{table}

\section{Summary}
We have presented the first algorithms that statically check the satisfaction of safety clauses in a smart space containing an arbitrary number of devices running an arbitrary number of concurrent routines, along with an expressive grammar for specifying such clauses. Knox checks safety properties {\it quickly} (without brute-force enumeration) and {\it accurately} (no missed violations, low false positive rates).
Our work opens up new future directions including
tackling dynamic versions of this problem, characterizing errors to further lower false positives, and considering additional error types, such as security violations.

\noindent {\bf Acknowledgments:} This work was supported in part by grants from NSF (CNS 1908888, CNS 2504595)  and the IBM-IL Discovery Accelerator Institute (IIDAI).

\printbibliography

@inproceedings {Ahsan01HotEdge,
author = {Shegufta Bakht Ahsan and Rui Yang and Shadi Abdollahian Noghabi and Indranil Gupta},
title = {Home, {SafeHome}: Ensuring a Safe and Reliable Home Using the Edge},
booktitle = {HotEdge 19},
year = {2019}
}

@ARTICLE{TapChecker,
  author={Chen, Liangyu and Wang, Chen and Chen, Cheng and Huang, Caidie and Chen, Xiaohong and Zhang, Min},
  journal={IEEE IoT Journal}, 
  title={{TapChecker}: A Lightweight {SMT}-Based Conflict Analysis for Trigger-Action Programming},
  year={2024},
  volume={11},
  number={12},
  pages={21411-21426},
  doi={10.1109/JIOT.2024.3374556}}

@INPROCEEDINGS{SmartDevicesStupid,
  author={He, Weijia and Martinez, Jesse and Padhi, Roshni and Zhang, Lefan and Ur, Blase},
  booktitle={IEEE SPW}, 
  title={When Smart Devices Are Stupid: Negative Experiences Using Home Smart Devices}, 
  year={2019},
  volume={},
  number={},
  pages={150-155},
  doi={10.1109/SPW.2019.00036}}

@inproceedings{EUDebug,
author = {Corno, Fulvio and De Russis, Luigi and Monge Roffarello, Alberto},
title = {Empowering End Users in Debugging Trigger-Action Rules},
year = {2019},
isbn = {9781450359702},
booktitle = {Proc. CHI},
pages = {1–13},
numpages = {13},
}

@inproceedings{BugsinTAP,
author = {Brackenbury, Will and Deora, Abhimanyu and Ritchey, Jillian and Vallee, Jason and He, Weijia and Wang, Guan and Littman, Michael L. and Ur, Blase},
title = {How Users Interpret Bugs in Trigger-Action Programming},
year = {2019},
isbn = {9781450359702},
doi = {10.1145/3290605.3300782},
booktitle = {Proc. CHI '19},
pages = {1–12},
numpages = {12}
}

@inproceedings{Ahsan02Eurosys,
author = {Ahsan, Shegufta B. and Yang, Rui and Noghabi, Shadi A. and Gupta, Indranil},
title = {Home, {SafeHome}: Smart Home Reliability with Visibility and Atomicity},
year = {2021},
isbn = {9781450383349},
doi = {10.1145/3447786.3456261},
booktitle = {Proc. EuroSys '21},
pages = {590–605},
numpages = {16},
location = {Online Event, United Kingdom}
}

@inproceedings {Hance,
author = {Travis Hance and Marijn Heule and Ruben Martins and Bryan Parno},
title = {Finding Invariants of Distributed Systems: It{\textquoteright}s a Small (Enough) World After All},
booktitle = {NSDI '21},
year = {2021},
isbn = {978-1-939133-21-2},
pages = {115--131},
month = apr,
}

@article{Hawblitzel,
author = {Hawblitzel, Chris and Howell, Jon and Kapritsos, Manos and Lorch, Jacob R. and Parno, Bryan and Roberts, Michael L. and Setty, Srinath and Zill, Brian},
title = {{IronFleet}: Proving Safety and Liveness of Practical Distributed Systems},
year = {2017},
volume = {60},
number = {7},
issn = {0001-0782},
doi = {10.1145/3068608},
journal = {CACM},
pages = {83–92},
numpages = {10}
}

@inproceedings{Liang,
author = {Liang, Chieh-Jan Mike and Karlsson, B\"{o}rje F. and Lane, Nicholas D. and Zhao, Feng and Zhang, Junbei and Pan, Zheyi and Li, Zhao and Yu, Yong},
title = {{SIFT}: Building an Internet of Safe Things},
year = {2015},
isbn = {9781450334754},
doi = {10.1145/2737095.2737115},
booktitle = {Proc. IPSN '15},
pages = {298–309},
numpages = {12}
}

@inproceedings{Zhou,
author = {Zhou, Qian and Ye, Fan},
title = {{APEX}: Automatic Precondition Execution with Isolation and Atomicity in Internet-of-Things},
year = {2019},
isbn = {9781450362832},
doi = {10.1145/3302505.3310066},
booktitle = {Proc. IoTDI '19},
pages = {25–36},
numpages = {12}
}

@inproceedings{Davidoff,
author = {Davidoff, Scott and Lee, Min Kyung and Yiu, Charles and Zimmerman, John and Dey, Anind K.},
title = {Principles of Smart Home Control},
year = {2006},
isbn = {9783540396345},
booktitle = {Proc. Ubicomp},
pages = {19–34},
numpages = {16}
}

@inproceedings{iotacalc,
author = {Newcomb, Julie L. and Chandra, Satish and Jeannin, Jean-Baptiste and Schlesinger, Cole and Sridharan, Manu},
title = {{IOTA}: a calculus for internet of things automation},
year = {2017},
isbn = {9781450355308},
doi = {10.1145/3133850.3133860},
booktitle = {Proc. SIGPLAN Onward! '17},
pages = {119–133},
numpages = {15}
}

@misc{IotA24,
author = "IoT-Analytics",
year = "2024",
title = "{IoT Startup Landscape 2024: 7 notable insights}",
url = "https://iot-analytics.com/iot-startup-landscape/",
month = jul,
lastaccessed = "May 2026",
}

@misc{IotA25,
author = "IoT-Analytics",
year = "2025",
title = "{State of IoT 2025: Number of connected IoT devices growing 14\% to 21.1 billion globally}",
url = "https://iot-analytics.com/number-connected-iot-devices/",
month = oct,
lastaccessed = "May 2026",
}

@misc{GVR26,
author = "Grand-View-Research",
year = "2026",
title = "Smart Home Market",
url = "https://www.grandviewresearch.com/industry-analysis/smart-homes-industry",
lastaccessed = "May 2026",
}

@inproceedings{DBLP:conf/tacas/BarbosaBBKLMMMN22,
  author    = {Haniel Barbosa and
               Clark W. Barrett and
               Martin Brain and
               Gereon Kremer and
               Hanna Lachnitt and
               Makai Mann and
               Abdalrhman Mohamed and
               Mudathir Mohamed and
               Aina Niemetz and
               Andres N{\"{o}}tzli and
               Alex Ozdemir and
               Mathias Preiner and
               Andrew Reynolds and
               Ying Sheng and
               Cesare Tinelli and
               Yoni Zohar},
  title     = {{cvc5}: {A} Versatile and Industrial-Strength {SMT} Solver},
  booktitle = {TACAS '22},
  series    = {Lecture Notes in Computer Science},
  volume    = {13243},
  pages     = {415--442},
  year      = {2022},
  doi       = {10.1007/978-3-030-99524-9_24},
  bibsource = {dblp computer science bibliography, https://dblp.org},
}

@misc{GoogleHome,
    author = "Google",
    title = "Google {Home}",
    url = "https://home.google.com/welcome/",
    year = "2024",
    month = jan,
}

@misc{AmazonAlexa,
    author = "Amazon",
    title = "Amazon {Alexa}",
    url = "https://developer.amazon.com/en-US/alexa",
    year = "2024",
    month = jan,
}

@misc{AppleHome,
    author = "Apple",
    title = "Apple {Home}",
    url = "https://www.apple.com/home-app/",
    year = "2024",
    month = jan,
}

@misc{SamsungSmartThings,
    author = "Samsung",
    title = "Samsung {SmartThings}",
    url = "https://www.smartthings.com/",
    year = "2024",
    month = jan,
}

@misc{HomeAssistant,
    author = "HomeAssistant",
    title = "{Home Assistant} Developer Docs",
    url = "https://developers.home-assistant.io/",
    year = "2024",
    month = oct,
}

@misc{IFTTT,
    author = "IFTTT",
    title = "{IFTTT}",
    url = "https://ifttt.com/",
    year = "2024",
    month = jan,
}

@misc{GoogleSafetyCritical,
    author = "Google",
    title = "Create and manage Routines for {Google} {Home} automations",
    url = "https://support.google.com/googlenest/answer/7029585",
    year = "2024",
    month = jan,
}

@misc{InternetOfShit,
    title = "{Internet of Shit}",
    author= "@internetofshit",
    url = "https://x.com/internetofshit",
    year = "2024",
    month = feb,
}

@misc{TermsSmartThings,
    author = "Samsung",
    title = "Samsung Service Terms and Conditions - {SmartThings} Services Supplement",
    url = "https://v3.account.samsung.com/policies/specials/smartthings.html",
    year = "2024",
    month = feb,
}

@misc{HomekitForum1,
    author = "MEGATOMI",
    title = "{HomeKit} Automations are broken in {iOS} 16",
    url = "https://discussions.apple.com/thread/254199872",
    year = "2022",
    month = sep,
}

@misc{HomekitForum2,
    author = "fabiom91",
    title = "{HomeKit} unreliable: unresponsive accessories",
    url = "https://discussions.apple.com/thread/253711255",
    year = "2022",
    month = mar,
}

@misc{HomekitForum3,
    author = "Mett03",
    title = "{HomeKit} automation not working after {iOS} 17",
    url = "https://discussions.apple.com/thread/255155940",
    year = "2023",
    month = sep,
}

@misc{HomekitForum4,
    author = "tomg15",
    title = "{HomeKit} Automations not working (again)",
    url = "https://discussions.apple.com/thread/254390066",
    year = "2022",
    month = nov,
}

@misc{SmartThingsForum1,
    author = "LowRange",
    title = "A warning not to rely on {SmartThings}",
    url = "https://community.smartthings.com/t/a-warning-not-to-rely-on-smartthings/218296",
    year = "2021",
    month = jan,
}

@article{DURAN2020100497,
title = {Programming and symbolic computation in {Maude}},
journal = {Journal of Logical and Algebraic Methods in Programming},
volume = {110},
pages = {100497},
year = {2020},
issn = {2352-2208},
doi = {10.1016/j.jlamp.2019.100497},
url = {https://www.sciencedirect.com/science/article/pii/S2352220818301135},
author = {Francisco Durán and Steven Eker and Santiago Escobar and Narciso Martí-Oliet and José Meseguer and Rubén Rubio and Carolyn Talcott}
}

\clearpage
\appendix
\knoxrestatingtrue
\label{appendix}

\section{NP Hardness: Single Routine Safety}\label{app:nphardsingle}

\NPHardSingle*

\begin{proof}
We prove for just one safety rule. Let \emph{SingleRoutineSafety} be an algorithm to solve the Single Routine Safety problem. Algorithm~\ref{algo1} shows the reduction from SAT:

\begin{algorithm}
\caption{\textbf{SAT Reduction to \emph{SingleRoutineSafety}}}
\label{algo1}
{\small
\KwIn{Boolean expression $B$}
\KwOut{``YES'' or ``NO''}
Define boolean $a$:
$a = $ Device $D$ in state $D_s$\;
Construct Safety Rule $S$ from expression $B$ and $a$:
$S = a \land \lnot (B)$\;
Construct routine $R$:
$R = \{\text{Set }D \text{ to } D_s\}$\;
Output ``YES'' iff \emph{SingleRoutineSafety}$(R, S)$ outputs ``NO''\;
Otherwise output ``NO''\;
}
\end{algorithm}

\emph{SingleRoutineSafety}$(R, S)$ returns ``YES'' iff there is \textit{no} setting that makes $\lnot (B)$ false, i.e.,  iff $B$ is not satisfiable. Due to this reduction, our problem is NP-hard. \end{proof}
\section{NP Hardness: Safe Initial States}\label{app:np-init}

\NPHardSafeInitialThm*

\begin{proof}
We again limit ourselves to one safety rule.
Let \emph{SingleRoutineSafety} be an algorithm to solve the Single Routine Safety problem, limited to safe initial states. Algorithm~\ref{algo:satreductionsafeinit} shows the reduction from SAT:

\begin{algorithm}
\caption{\bf SAT Reduction to \emph{SingleRoutineSafety} with safe initial states}
\label{algo:satreductionsafeinit}
\KwIn{Boolean expression $B$}
\KwOut{``YES'' or ``NO''}
Define boolean $a$:
$a = $ Device $D$ in state $D_s$\;
Construct Safety Rule $S$ from expression $B$ and $a$:
$S = a \lor \lnot (B)$\;
Construct routine $R$:
$R = \{\text{Set Device } D \text{ to state } D_s', D_s \neq D_s'\}$\;
Output ``YES'' iff \emph{SingleRoutineSafety}$(R, S)$ outputs ``NO''\;
Otherwise output ``NO''\;
\end{algorithm}

If $B$ is not satisfiable, then $S$ is trivially safe. However, if $B$ \textit{is} satisfiable, we can choose such a setting, with $D$ in state $D_s$ ($a=\text{True}$) as our safe starting state. But entering state $D_s'$ (violating $a$) will violate safety! Therefore, \emph{SingleRoutineSafety}$(R, S)$ returns ``YES'' iff $B$ is not satisfiable, and the problem is NP-hard.
\end{proof}
\section{One-To-One Contribution}\label{app:one-to-one-contribution}

\onetoonecontribution*

\begin{proof}

By induction.

\noindent\textbf{Base Case}: $n$ is the root, and trivially satisfies the theorem by Definition \ref{defn::RootContribution}.

\noindent\textbf{Inductive Case}: Assume $n$'s parent's setting can contribute to {\it only} rule satisfaction or rule violation. It is sufficient to show that $n$'s setting can only contribute to one truth value for its parent, which then contributes to satisfaction or violation. We have five possibilities:

\squishlist
\item\textbf{$n$'s parent is a NOT node}: The claim is trivially satisfied.

\item\textbf{$n$'s parent is an OR node}: If $ST(n) = \textit{false}$, then setting $ST(n) = \textit{true}$ trivially satisfies the parent, no matter  the other sibling values of $n$.

Likewise, if all other siblings are  $\textit{false}$, then setting $ST(n) = \textit{false}$ violates the parent, and there exists no setting of siblings such that setting $ST(n)=\textit{false}$  changes parent's value from $\textit{false}$ to $\textit{true}$.

\item\textbf{$n$'s parent is an AND node}: If all siblings are set to $\textit{true}$, then setting $ST(n)$ to $\textit{true}$ satisfies the parent, and there exists no setting of other siblings such that setting $ST(n)=\textit{true}$  changes the parent's value from $\textit{true}$ to $\textit{false}$.

Likewise, if $ST(n) = \textit{true}$, then setting $ST(n)=\textit{false}$ trivially violates the parent, no matter $n$'s other sibling values.

\item\textbf{$n$'s parent is AT MOST $k$}: If $k$ siblings of $n$ are set to $\textit{true}$, then setting $ST(n)$ to $\textit{true}$ violates the parent. There exists no setting of $n$'s siblings such that setting $ST(n) = \textit{true}$ changes the parent's value from $\textit{false}$ to $\textit{true}$.

Likewise, if $k$ siblings and $n$ are $\textit{true}$, then setting $ST(n)$ to $\textit{false}$ satisfies the parent. There is no setting of $n$'s siblings such that setting $n$ to $\textit{false}$ changes the parent's value from $\textit{true}$ to $\textit{false}$.

\item\textbf{$n$'s parent is AT LEAST $k$}: If $k - 1$ siblings and $n$ are set to $\textit{true}$, then setting $ST(n)$ to $\textit{false}$ violates the parent. There is no setting to $n$'s siblings such that setting $ST(n)$ to $\textit{false}$ changes the parent's value from $\textit{false}$ to $\textit{true}$.

Likewise, if $k - 1$ siblings are set to $\textit{true}$, then setting $n$ to $\textit{true}$ satisfies the parent. There is no setting to $n$'s siblings such that setting $n$ to $\textit{true}$ changes the parent's value from $\textit{true}$ to $\textit{false}$.
\squishend

Hence, each node's setting can contribute only a single truth value to its parent, and by extension to exactly one of rule satisfaction or violation. Additionally, this result shows that a setting contributing to satisfaction is the opposite of the corresponding setting for violation.
 
By induction, these properties hold for the entire tree. 
\end{proof}
\section{Correctness of Routine Safety}\label{app:correctness}

\correctness*

\begin{proof}
We define $C1$'s location vector as $<m_1, ..., m_{|\mathcal{R}|}>$ and $C2$'s location vector as $<i_1, ..., i_{|\mathcal{R}|}>$, and proceed with a proof by contradiction.

Assume a safety rule $S$ is violated at some arbitrary location vector $t = <j_1, ..., j_{|\mathcal{R}|}>$, with $m_k \leq j_k \leq i_k$ for all $1 \leq k \leq |\mathcal{R}|$. Let $C'$ be the result of merging cuts taken from each routine $R_k$ at location $j_k$, assumed to be safe. 
Let the ordered leaf states for $C'$ and $t$ be $L_{C'}$ and $L_t$, respectively.
Let $\Delta s$ be the set of leaves in $S$'s tree, such that for every leaf $l \in \Delta s$, $l$'s state in $L_t$ and $l$'s state in $L_{C'}$ differ. Note  $\Delta s \neq \varnothing$ iff $L_{C'}$ and $L_t$ disagree on at least one leaf state.

For an arbitrary $l \in \Delta s$, consider the cause of the disagreement. By definition, cuts arising from multiple routines preserve the exact state recorded from constituent routines in a state list, with two notable exceptions. A merge can consist of routines that 1. never set that device state, in which case the cut records \textit{unset},  or 2. disagree on a device state, in which case the cut records \textit{conflict}. In both cases,  Cut Contextualizer assumes an attack state for the leaf in question.

$L_t$, on the other hand, is exactly a record of leaf states taken at time $t$. So every $l$ must have the property of being an attack in $L_{C'}$ and a defense in $L_t$.
Next we construct $L_t$ from $L_{C'}$. From $\Delta s$'s definition, this transformation is merely changing the state of every $l \in \Delta s$. If every $l$ was an attack in $L_{C'}$, then it must be a defense in $L_t$. Since $L_{C'}$ satisfies $S$ by assumption, $L_t$ must also satisfy $S$. Thus, such a location vector $t$ (violating  $S$) cannot exist.

\end{proof}
\section{Single-Routine False Positive Avoidance}\label{app:singleroutinefpavoidance}

\singleroutinefp*

\begin{proof}
If no false positives exist, then an unsafe single routine cut $C$ exists if and only if a point of execution exists where the routine violates safety for some initial state.
Theorem~\ref{thm::Merging}'s result implies the contrapositive:
if there is an unsafe interleaving, there must be an unsafe cut.
Without loss of generality we only consider devices present in $T$, and consider $C$ as a subroutine executing $C$'s location number of commands. Note that when considering a single routine, a cut state of \textit{conflict} is not possible; only leaves that are \textit{unset} are not explicitly known.
Further, since there is at most one leaf for each device in $T$, there are no contradictory values (Section~\ref{sec:attackdefense}).
We construct initial home state $h$: for each device in $T$, we set it to its attack state. Since \textit{unset} leaves are by definition not modified by the subroutine, these leaves must still be in their attack state. Thus, the home state after executing the subroutine is \textit{exactly} the state that the Cut Contextualizer constructs. Since stateful trees are direct evaluations of safety, if there is an unsafe cut, there is an unsafe interleaving.
\end{proof}
\section{Multi-Routine False Positive Avoidance}\label{app:multiroutinefpavoidance}

\multiroutinefp*

\begin{proof}
We follow the thinking of Theorem~\ref{SingleRoutineFalsePositives} for multi-routine cuts.
For the first condition pair, since there is at most one leaf per device, every command that acts on that device would set all relevant leaves to their attack or defense states, so no commands are ambiguous. For leaves in \textit{conflict}, when choosing an interleaving, order their respective commands so that all routines execute their defense commands first, and then all routines execute their attack commands. We can do this since all routines execute attacks \textit{after} defenses for conflict leaves (by assumption). 
 
For condition (B), assuming a single \textit{conflict} leaf $l$, there must be at least one routine that set it to its attack state and at least one other that set it to its defense state. Without loss of generality, let routine $R_1$ be one of the former and routine $R_2$ be one of the latter. An unsafe interleaving will run all routines other than $R_1$ and $R_2$ in any arbitrary order. Then, when running $R_1$ and $R_2$ afterwards, execute $R_2$'s defense command first, and $R_1$'s attack command second.
 
In both cases, since all other devices are explicitly known by $C$'s location vector, they do not need to be ordered. All conflict leaves in the selected interleaving are by construction in their attack state. 
Thus, the home state after executing this interleaving has \textit{exactly} the state that the Cut Contextualizer constructs. Since stateful trees are direct evaluations of safety, there is an unsafe cut if and only if there is an unsafe interleaving. 
\end{proof}
\section{Attack/Defense from a Safe Start}\label{app:safestart}

\SafeStart*

\begin{proof}
To use Theorem~\ref{thm::AttackDefense} we need the immediate predecessor to be safe. There are two cases--- either it is the cut before all routine commands (assumed safe), or it is preceded by an attack or ambiguous command in which case it is relevant and checked for safety.
\end{proof}
\section{Attack/Defense Runtime Analysis}
\label{app:AttackDefenseAnalysis}

\AttackDefenseRuntimeThm*

\begin{proof}
Command Classification requires every unique command across all routines to be categorized against a rule. The classification is remade for every rule. Based on Section \ref{CoreSystemAnalysis}'s variables, if the worst case total unique command count is $I = \Sigma_{R_i \in \mathcal{R}} |R_i|$ and  total  leaves across all safety rules is $L_\mathcal{S}$, then Command Classification takes $O(I \cdot L_\mathcal{S})$ time.

The total number of multi-routine cuts to be checked for safety depends on the product of number of single routine cuts from each routine, and  total routine count. Theorem \ref{thm::AttackDefense} can reduce the first term. %
A worst-case scenario is a rule set $\mathcal{S}$ and a routine set $\mathcal{R}$ where every command in the routine set is classified as ambiguous for every safety rule in $\mathcal{S}$. Each multi-routine cut  in Section \ref{subsec::Mlirad} must be checked for safety, requiring %
$O(N_\mathcal{S} + K \cdot L_\mathcal{S} + |R_{max}|^{|\mathcal{R}|} \cdot (L_\mathcal{S} + N_\mathcal{S}) + I \cdot L_\mathcal{S})$.

This can be reduced if $\mathcal{S}$ and $\mathcal{R}$ are restricted such that no commands in $\mathcal{R}$ are classified as ambiguous across all rules in $\mathcal{S}$. Routines should start with an attack command and then \emph{alternate} between defenses and attacks — the only irrelevant cuts are those between a defense command and the following attack command (or after the last command if a defense). 
Thus total locations needed for single routine checking, $K$, can be reduced to a $K_{ad}$. Since we halve the number of cuts to be checked and ignore the starting cut, $K_{ad} = \sum_{R_i\in\mathcal{R}} \frac{|R_i|}{2}$. The same reasoning means that the number of location vectors as multi-routine cuts  that are needed to ensure concurrent routine safety are also reduced: $\left(\frac{|R_{max}|}{2}\right)^{|\mathcal{R}|}$. Therefore, the adjusted worst case runtime is $O(N_\mathcal{S} + K_{ad} \cdot L_\mathcal{S} + K_{ad} \cdot N_{\mathcal{S}} + \left(\frac{|R_{max}|}{2}\right)^{|\mathcal{R}|} \cdot (L_\mathcal{S} + N_\mathcal{S})$).
\end{proof}

\section{Wall Equivalence}
\label{app:walleq}

\WallEqThm*

\begin{proof}
Consider leaves {\it outside} the wall. By definition, these must either be  \textit{conflict} cut state or  attack state, both of which our Cut Contextualizer sets to attack state. We assign leaves \textit{inside} the wall (by definition) the same state in $L_1$ \& $L_2$. Hence $\forall$ leaves: $L_1$ = $L_2$.
\end{proof}
\section{Wall Transitivity}
\label{app:WallTransitivityApp}

\WallTransitivityThm*

\begin{proof}
Since Table~\ref{table::cut-state-merge}'s operations are commutative, let $M'$ siege $W_M$. Then, $W_{M'} \subseteq W_M$. $M'$ is safe by assumption, sieging $W_M$ with $C$ removes no leaves, and $W_{M'} \subseteq W_M$, so sieging $W_{M'}$ with $C$ also removes no leaves.

By Theorem~\ref{thm::siege1.2}, the leaves both inside and outside the wall take the states that the Cut Contextualizer assumes. Therefore, if a \emph{safe} multi routine cut's wall remains unchanged after merging a new single routine cut, the result is still safe.
Since $M$ and $M'$ are each safe, merging $C$ is also safe as the wall remains unchanged. 
\end{proof}

\section{Wall/Siege Runtime Analysis}
\label{app:WallSiegeRuntimeAnalysis}

\ConSafRuntimeWallThm*

\begin{proof}
Safety evaluation is needed only when a leaf is removed from the wall. In the worst case, each non-zero component of a multi-routine cut location vector removes exactly 1 leaf when comparing walls (else we were still safe). Once a wall is empty, further cuts are trivially unsafe.
Constructing all multi-routine cuts takes $O(V\cdot L_\mathcal{S})$ time (Section~\ref{CoreSystemAnalysis}), where $V=\prod_{R_i\in\mathcal{R}} (|R_i| + 1)$. %
The complexity comes from number of location vectors with exactly $L_\mathcal{S}$ non-zero components:  $\binom{|\mathcal{R}|}{L_\mathcal{S}} \cdot |R_{max}|^{L_\mathcal{S}}$, where $R_{max}$ has the maximum commands of any routine. This gives a runtime upper bound: 
$$
{\textstyle O(V\cdot L_\mathcal{S} + \binom{|\mathcal{R}|}{L_\mathcal{S}} \cdot |R_{max}|^{L_\mathcal{S}} \cdot N_\mathcal{S})}
$$
\end{proof}

\end{document}